\documentclass{llncs}

\usepackage{times}
\usepackage{xcolor}
\usepackage{soul}
\usepackage{syntax}
\usepackage[utf8]{inputenc}

\usepackage{enumerate,paralist}
\usepackage{parskip}
\usepackage{multirow}
\usepackage{graphicx}
\usepackage{graphics}
\usepackage{color}
\usepackage[T1]{fontenc}
\usepackage{textcomp}
\usepackage{listing}
\usepackage{listings}
\usepackage{tikz,pgffor}
\usetikzlibrary{arrows}
\usetikzlibrary{shapes}
\usepackage{mathtools}

\usetikzlibrary{calc}
\usetikzlibrary{automata}
\usetikzlibrary{positioning}

\tikzstyle{proglabel}=[shape=circle,draw,inner sep=0pt,minimum size=5mm]
\tikzstyle{tran}=[draw,->,>=stealth, rounded corners]

\usepackage{amsmath}
\usepackage{amssymb}
\usepackage{stmaryrd}
\usepackage{hyperref}
\usepackage{listings}
\lstdefinelanguage{prog}
{
morekeywords={prob, if, then, else, fi, while, do, od, true, false, and, or, skip},
sensitive = false
}

\newcommand{\Rset}{\mathbb{R}}
\newcommand{\Nset}{\mathbb{N}}
\newcommand{\Zset}{\mathbb{Z}}

\newcommand{\probm}{\mathbb{P}}
\newcommand{\expv}{\mathbb{E}}
\newcommand{\var}{\mathbf{Var}}
\newcommand{\pv}{\mathbf{v}}
\newcommand{\rv}{\mathbf{u}}

\newcommand{\condexpv}[2]{{\expv}{\left({#1}{\mid}{#2}\right)}}

\newcommand{\stopping}[1]{Z_{#1}}

\newcommand{\dpd}{q}
\newcommand{\supp}[1]{{\mathrm{supp}}{\left(#1\right)}}

\title{New Approaches for Almost-Sure Termination of Probabilistic Programs
}

\author{Mingzhang Huang\inst{1}
\and Hongfei Fu \inst{2}
\and Krishnendu Chatterjee\inst{3}}

\institute{BASICS Lab, Shanghai Jiao Tong University, China\\
  \email{mingzhanghuang@gmail.com}\and
Shanghai Jiao Tong University, China\\
  \email{fuhf@cs.sjtu.edu.cn}
\and IST Austria, Austria \\
  \email{krishnendu.chatterjee@ist.ac.at}}

\usepackage{tabularx}
\newcolumntype{L}[1]{>{\raggedright\arraybackslash}p{#1}}
\newcolumntype{C}[1]{>{\centering\arraybackslash}p{#1}}
\newcolumntype{R}[1]{>{\raggedleft\arraybackslash}p{#1}}

\begin{document}

\maketitle

\begin{abstract}
We study the almost-sure termination problem for probabilistic programs.
First, we show that supermartingales with lower bounds on conditional absolute difference provide
a sound approach for the almost-sure termination problem.
Moreover, using this approach we can obtain explicit optimal bounds on tail probabilities of non-termination
within a given number of steps.
Second, we present a new approach based on Central Limit Theorem for the almost-sure termination problem, and
show that this approach can establish almost-sure termination of programs which none of the existing approaches can handle.
Finally, we discuss algorithmic approaches for the two above methods that lead to automated
analysis techniques for almost-sure termination of probabilistic programs.
\end{abstract}


\section{Introduction}\label{sec:introduction}

\noindent{\em Probabilistic Programs.}
Probabilistic programs are classical imperative programs extended with \emph{random value
generators} that produce random values according to some desired probability
distribution~\cite{DBLP:conf/aistats/MeentYMW15,DBLP:conf/concur/Yang17,DBLP:conf/icse/GordonHNR14}.
They provide the appropriate model for a wider
variety of applications, such as analysis of stochastic network
protocols~\cite{BaierBook,prism}, robot planning~\cite{kaelbling1998planning}, etc.
General probabilistic programs induce infinite-state Markov processes with complex behaviours, so that the
formal analysis is needed in critical situations.
The formal analysis of probabilistic programs is
an active research topic across different disciplines, such as
probability theory and statistics~\cite{Kemeny,Rabin63,PazBook},
formal methods~\cite{BaierBook,prism},
artificial intelligence~\cite{LearningSurvey},
and programming languages~\cite{SriramCAV,HolgerPOPL,SumitPLDI,EGK12,DBLP:conf/popl/ChatterjeeFNH16}.

\noindent{\em Termination Problems.} In this paper, we focus on proving termination properties of probabilistic programs.
Termination is the most basic and fundamental notion of liveness for programs.
For non-probabilistic programs, the proof of termination coincides with
the construction of \emph{ranking functions}~\cite{rwfloyd1967programs},
and many different approaches exist for such construction~\cite{DBLP:conf/cav/BradleyMS05,DBLP:conf/tacas/ColonS01,DBLP:conf/vmcai/PodelskiR04,DBLP:conf/pods/SohnG91}.
For probabilistic programs the most natural and basic extensions of the
termination problem are \emph{almost-sure} termination and \emph{finite} termination.
First, the almost-sure termination problem asks whether the program
terminates with probability~1.
Second, the finite termination problem asks whether the expected
termination time is finite.
Finite termination implies almost-sure termination, while the converse
is not true in general.
Here we focus on the almost-sure termination problem.


\noindent{\em Previous Results.}
Below we describe the most relevant previous results on termination of
probabilistic programs.
\begin{compactitem}
\item \emph{finite probabilistic choices.}
First,  quantitative invariants were used in~\cite{MM04,MM05} to establish termination
for probabilistic programs with non-determinism, but restricted only to finite probabilistic choices.

\item \emph{infinite probabilistic choices without non-determinism.}
The approach in~\cite{MM04,MM05} was extended in~\cite{SriramCAV} to \emph{ranking
supermartingales} to obtain a sound (but not complete) approach
for almost-sure termination over infinite-state probabilistic programs with
infinite-domain random variables, but without non-determinism.
For countable state space probabilistic programs without non-determinism,
 the Lyapunov ranking functions provide a sound and complete
method to prove finite termination~\cite{BG05,Foster53}.

\item \emph{infinite probabilistic choices with non-determinism.}
In the presence of non-determinism, the Lyapunov-ranking-function method
as well as the ranking-supermartingale method are sound but not complete~\cite{HolgerPOPL}.
Different approaches based on martingales and proof rules have been studied
for finite termination~\cite{HolgerPOPL,DBLP:conf/esop/KaminskiKMO16}.
The synthesis of linear and polynomial ranking supermartingales have been
established~\cite{DBLP:conf/popl/ChatterjeeFNH16,DBLP:conf/cav/ChatterjeeFG16}.
Approaches for high-probability termination and non-termination has also been
considered~\cite{ChatterjeeNZ2017}.
Recently, supermartingales and lexicographic ranking supermartingales have been
considered for proving almost-sure termination of probabilistic programs~\cite{MclverMorgan2016,DBLP:journals/pacmpl/AgrawalC018}.

\end{compactitem}
Note that the problem of deciding termination of probabilistic programs is undecidable~\cite{DBLP:conf/mfcs/KaminskiK15}, and its precise
undecidability characterization has been investigated.
Finite termination of recursive probabilistic programs has also been studied through proof rules~\cite{DBLP:conf/lics/OlmedoKKM16}.

\noindent{\em Our Contributions.} Now we formally describe our contributions.
We consider probabilistic programs where all program variables are integer-valued.
Our main contributions are three folds.

\begin{compactitem}

\item {\em Almost-Sure Termination: Supermartingale-Based Approach.}
We show new results that supermartingales (i.e., not necessarily ranking supermartingales) with lower bounds on
conditional absolute difference present a sound approach for proving almost-sure termination of probabilistic programs.
Moreover, no previous supermartingale based approaches present explicit (optimal) bounds on
tail probabilities of non-termination within a given number of steps.

\item {\em Almost-Sure Termination: CLT-Based Approach.}
We present a new approach based on Central Limit Theorem (CLT) that is sound to establish
almost-sure termination.
The extra power of CLT allows one to prove probabilistic programs where no global lower bound exists for values of program variables,
while previous approaches based on (ranking) supermartingales~\cite{HolgerPOPL,DBLP:conf/popl/ChatterjeeFNH16,DBLP:conf/cav/ChatterjeeFG16,MclverMorgan2016,DBLP:journals/pacmpl/AgrawalC018}.
For example, when we consider the program $\textbf{while}~n\ge 1~\textbf{ do } n:=n+r\textbf{ od}$ and
take the sampling variable $r$ to observe the probability distribution $\probm$ such that $\probm(r=k)=\frac{1}{2^{|k|+1}}$ for all integers $k\neq 0$, then the value of $n$
could not be bounded from below during program execution; previous approaches fail on this example, while our CLT-based approach succeeds.

\item {\em Algorithmic Methods.}
We discuss algorithmic methods for the two approaches we present, showing that we not only
present general approaches for almost-sure termination, but possible automated analysis techniques as well.

\end{compactitem}


\noindent{\em Recent Related Work.}
In the recent work~\cite{MclverMorgan2016}, supermartingales are also considered for proving almost-sure termination.
The difference between our results and the paper are as follows.
First, while the paper relaxes our conditions to obtain a more general result on almost-sure termination,
our supermartingale-based approach can derive optimal tail bounds along with proving almost-sure termination.
Second, our CLT-based approach can handle programs without lower bound on values of program variables, while the result in the paper
requires a lower bound.
We also note that our supermartingale-based results are independent of the paper
(see arXiv versions~\cite{DBLP:journals/corr/McIverM16} and ~\cite[Theorem 5 and Theorem 6]{DBLP:journals/corr/ChatterjeeF17}).
A more elaborate description of related works is put in Section~\ref{sect:realtedwork}.



\newcommand{\AS}{Q}

\section{Preliminaries}\label{sect:preliminaries}

Below we first introduce some basic notations and concepts in probability theory (see e.g. the standard textbook
\cite{probabilitycambridge} for details), then present the syntax and semantics of our probabilistic programs.

\subsection{Basic Notations and Concepts}\label{sect:basicnotations}

In the whole paper, we use $\Nset$, $\Nset_0$, $\Zset$, and $\Rset$ to denote the sets of all positive integers, non-negative integers, integers, and real numbers, respectively.

\noindent{\bf Probability Space.} A \emph{probability space} is a triple $(\Omega,\mathcal{F},\probm)$, where $\Omega$ is a non-empty set (so-called \emph{sample space}), $\mathcal{F}$ is a \emph{$\sigma$-algebra} over $\Omega$ (i.e., a collection of subsets of $\Omega$ that contains the empty set $\emptyset$ and is closed under complementation and countable union) and $\probm$ is a \emph{probability measure} on $\mathcal{F}$, i.e., a function $\probm\colon \mathcal{F}\rightarrow[0,1]$ such that (i) $\probm(\Omega)=1$ and
(ii) for all set-sequences $A_1,A_2,\dots \in \mathcal{F}$ that are pairwise-disjoint
(i.e., $A_i \cap A_j = \emptyset$ whenever $i\ne j$)
it holds that $\sum_{i=1}^{\infty}\probm(A_i)=\probm\left(\bigcup_{i=1}^{\infty} A_i\right)$.
Elements of $\mathcal{F}$ are usually called \emph{events}.
We say an event $A\in\mathcal{F}$ holds \emph{almost-surely} (a.s.) if $\probm(A)=1$.

\noindent{\bf Random Variables.}~\cite[Chapter 1]{probabilitycambridge} A \emph{random variable} $X$ from a probability space $(\Omega,\mathcal{F},\probm)$
is an $\mathcal{F}$-measurable function $X\colon \Omega \rightarrow \Rset \cup \{-\infty,+\infty\}$, i.e.,
a function satisfying the condition that for all $d\in \Rset \cup \{-\infty,+\infty\}$, the set $\{\omega\in \Omega\mid X(\omega)<d\}$ belongs to $\mathcal{F}$; $X$ is \emph{bounded} if there exists a real number $M>0$ such that for all $\omega\in\Omega$, we have $X(\omega)\in\Rset$ and $|X(\omega)|\le M$.
By convention, we abbreviate $+\infty$ as $\infty$.

\noindent{\bf Expectation.} The \emph{expected value} of a random variable $X$ from a probability space $(\Omega,\mathcal{F},\probm)$, denoted by $\expv(X)$, is defined as the Lebesgue integral of $X$ w.r.t $\probm$, i.e.,
$\expv(X):=\int X\,\mathrm{d}\probm$~;
the precise definition of Lebesgue integral is somewhat technical and is
omitted  here (cf.~\cite[Chapter 5]{probabilitycambridge} for a formal definition).
In the case that the range of $X$, $\mbox{\sl ran}~X=\{d_0,d_1,\dots,d_k,\dots\}$ is countable with distinct $d_k$'s, we have
$\expv(X)=\sum_{k=0}^\infty d_k\cdot \probm(X=d_k)$.

\noindent{\bf Characteristic Random Variables.} Given random variables $X_0,\dots,X_n$ from a probability space $(\Omega,\mathcal{F},\probm)$ and a predicate $\Phi$ over $\Rset \cup \{-\infty,+\infty\}$, we denote by $\mathbf{1}_{\phi(X_0,\dots,X_n)}$ the random variable such that
$\mathbf{1}_{\phi(X_0,\dots,X_n)}(\omega)=1$ if $\phi\left(X_0(\omega),\dots,X_n(\omega)\right)$ holds, and
$\mathbf{1}_{\phi(X_0,\dots,X_n)}(\omega)=0$ otherwise.

By definition, $\expv\left(\mathbf{1}_{\phi(X_0,\dots,X_n)}\right)=\probm\left(\phi(X_0,\dots,X_n)\right)$.
Note that if $\phi$ does not involve any random variable, then $\mathbf{1}_{\phi}$ can be deemed as a constant whose value depends only on whether $\phi$ holds or not.

\noindent{\bf Filtrations
 and Stopping Times.
 } A \emph{filtration} of a probability space $(\Omega,\mathcal{F},\probm)$ is an infinite sequence $\{\mathcal{F}_n \}_{n\in\Nset_0}$ of $\sigma$-algebras over $\Omega$ such that $\mathcal{F}_n \subseteq \mathcal{F}_{n+1} \subseteq\mathcal{F}$ for all $n\in\Nset_0$.
A \emph{stopping time} (from $(\Omega,\mathcal{F},\probm)$) w.r.t $\{\mathcal{F}_n\}_{n\in\Nset_0}$ is a random variable $R:\Omega\rightarrow \Nset_0\cup\{\infty\}$ such that for every $n\in\Nset_0$, the event $R\le n$ belongs to $\mathcal{F}_n$.

\noindent{\bf Conditional Expectation.}
Let $X$ be any random variable from a probability space $(\Omega, \mathcal{F},\probm)$ such that $\expv(|X|)<\infty$.
Then given any $\sigma$-algebra $\mathcal{G}\subseteq\mathcal{F}$, there exists a random variable (from $(\Omega, \mathcal{F},\probm)$), conventionally denoted by $\condexpv{X}{\mathcal{G}}$, such that
\begin{compactitem}
\item[(E1)] $\condexpv{X}{\mathcal{G}}$ is $\mathcal{G}$-measurable, and
\item[(E2)] $\expv\left(\left|\condexpv{X}{\mathcal{G}}\right|\right)<\infty$, and
\item[(E3)] for all $A\in\mathcal{G}$, we have $\int_A \condexpv{X}{\mathcal{G}}\,\mathrm{d}\probm=\int_A {X}\,\mathrm{d}\probm$.
\end{compactitem}
The random variable $\condexpv{X}{\mathcal{G}}$ is called the \emph{conditional expectation} of $X$ given $\mathcal{G}$.
The random variable $\condexpv{X}{\mathcal{G}}$ is a.s. unique in the sense that if $Y$ is another random variable satisfying (E1)--(E3), then $\probm(Y=\condexpv{X}{\mathcal{G}})=1$.

\noindent{\bf Discrete-Time Stochastic Processes.}
A \emph{discrete-time stochastic process} is a sequence $\Gamma=\{X_n\}_{n\in\Nset_0}$ of random variables where $X_n$'s are all from some probability space (say, $(\Omega,\mathcal{F},\probm)$);
and $\Gamma$ is \emph{adapted to} a filtration $\{\mathcal{F}_n\}_{n\in\Nset_0}$ of sub-$\sigma$-algebras of $\mathcal{F}$ if for all $n\in\Nset_0$, $X_n$ is $\mathcal{F}_n$-measurable.

\noindent{\bf  Difference-Boundedness.}
A discrete-time stochastic process $\Gamma=\{X_n\}_{n\in\Nset_0}$  is \emph{difference-bounded} if there is $c\in(0,\infty)$ such that for all $n\in\Nset_0,\ $ $|X_{n+1}-X_n|\le c$ a.s..

\noindent{\bf Stopping Time $\stopping{\Gamma}$.}
Given a discrete-time stochastic process $\Gamma=\{X_n\}_{n\in\Nset_0}$ adapted to a filtration $\{\mathcal{F}_n\}_{n\in\Nset_0}$, we define the random variable $\stopping{\Gamma}$ by
$\stopping{\Gamma}(\omega):=\min\{n\mid X_{n}(\omega)\le 0\}$
where $\min\emptyset:=\infty$.
By definition, $\stopping{\Gamma}$ is a stopping time w.r.t $\{\mathcal{F}_n\}_{n\in\Nset_0}$.

\noindent{\bf Martingales.} A discrete-time stochastic process $\Gamma=\{X_n\}_{n\in\Nset_0}$ adapted to a filtration $\{\mathcal{F}_n\}_{n\in\Nset_0}$ is a \emph{martingale} (resp. \emph{supermartingale})
if for every $n\in\Nset_0$, $\expv(|X_n|)<\infty$ and it holds a.s. that
$\condexpv{X_{n+1}}{\mathcal{F}_n}=X_n$ (\mbox{resp. } $\condexpv{X_{n+1}}{\mathcal{F}_n}\le X_n$).
We refer to ~\cite[Chapter~10]{probabilitycambridge} for more details.

\noindent{\bf Discrete Probability Distributions over Countable Support.} A \emph{discrete probability distribution} over a countable set $U$ is a function $\dpd:U\rightarrow[0,1]$ such that $\sum_{z\in U}\dpd(z)=1$.
The \emph{support} of $q$, is defined as $\supp{q}:=\{z\in U\mid q(z)>0\}$.


\subsection{The Syntax and Semantics for Probabilistic Programs}

In the sequel, we fix two countable sets, the set of \emph{program variables}  and the set of \emph{sampling variables}.
W.l.o.g, these two sets are disjoint. Informally, program variables are the variables that are directly related to the control-flow and the data-flow of a program, while sampling variables reflect randomized inputs to programs.
In this paper, we consider integer-valued variables, i.e., every program variable holds an integer upon instantiation,
while every sampling variable is bound to a discrete probability distribution over integers. 
Possible extensions to real-valued variables are discussed in Section~\ref{sect:algorithmicmethods}.

\noindent{\bf The Syntax.} The syntax of probabilistic programs is illustrated by the grammar in Figure~\ref{fig:syntax}.
Below we explain the grammar.
\begin{compactitem}
\item \emph{Variables.} Expressions $\langle\mathit{pvar}\rangle$ (resp. $\langle\mathit{rvar}\rangle$) range over program (resp. sampling) variables.
\item \emph{Arithmetic Expressions.} Expressions $\langle\mathit{expr}\rangle$ (resp. $\langle\mathit{pexpr}\rangle$) range over arithmetic expressions over both program and sampling variables (resp. program variables), respectively. As a theoretical paper, we do not fix the detailed syntax for $\langle\mathit{expr}\rangle$ and $\langle\mathit{pexpr}\rangle$.
\item \emph{Boolean Expressions.} Expressions $\langle\mathit{bexpr}\rangle$ range over propositional arithmetic predicates over program variables.
\item \emph{Programs.} A program from $\langle \mathit{prog}\rangle$ could be either an assignment statement indicated by `$:=$', or `\textbf{skip}' which is the statement that does nothing, or a conditional branch indicated by the keyword `\textbf{if}', or a while-loop indicated by the keyword `\textbf{while}', or a sequential composition of statements connected by semicolon.
\end{compactitem}

\begin{figure}
\[\begin{array}{rrl}
\langle \mathit{prog}\rangle &::=&  \mbox{`\textbf{skip}'}\\
&\mid& \langle\mathit{pvar}\rangle \,\mbox{`$:=$'}\, \langle\mathit{expr} \rangle\\
&\mid&  \langle\mathit{prog}\rangle \, \text{`;'} \langle\mathit{prog}\rangle\\
& \mid& \mbox{`\textbf{if}'} \, \langle\mathit{bexpr}\rangle\,\mbox{`\textbf{then}'} \, \langle \mathit{prog}\rangle \, \mbox{`\textbf{else}'} \, \langle \mathit{prog}\rangle \,\mbox{`\textbf{fi}'}\\
&\mid&  \mbox{`\textbf{while}'}\, \langle\mathit{bexpr}\rangle \, \text{`\textbf{do}'} \, \langle \mathit{prog}\rangle \, \text{`\textbf{od}'}
\\
\langle\mathit{literal} \rangle &::=& \langle\mathit{pexpr} \rangle\, \mbox{`$\leq$'} \,\langle\mathit{pexpr} \rangle \mid \langle\mathit{pexpr} \rangle\, \mbox{`$\geq$'} \,\langle\mathit{pexpr} \rangle
\\
\langle \mathit{bexpr}\rangle &::=&  \langle \mathit{literal} \rangle \mid \neg \langle\mathit{bexpr}\rangle \mid \langle \mathit{bexpr} \rangle \, \mbox{`\textbf{or}'} \, \langle\mathit{bexpr}\rangle
\mid \langle \mathit{bexpr} \rangle \, \mbox{`\textbf{and}'} \, \langle\mathit{bexpr}\rangle
\end{array}\]
\caption{The Syntax of Probabilistic Programs}
\label{fig:syntax}
\end{figure}


\begin{remark}\label{rmk:icsefosegordonHNR14}
The syntax of our programming language is quite general and covers major features of probabilistic programming. For example, compared with a popular probabilistic-programming language from~\cite{DBLP:conf/icse/GordonHNR14},
the only difference between our syntax and theirs is that they have extra observe statements.\qed
\end{remark}

\noindent{\bf Single (Probabilistic) While Loops.} In order to develop approaches for proving almost-sure termination of probabilistic programs, we first analyze the almost-sure termination of programs with a single while loop. Then, we demonstrate that the almost-sure termination of general probabilistic programs without nested loops can be obtained by the almost-sure termination of all components which are single while loops and loop-free statements (see Section~\ref{sect:algorithmicmethods}).
Formally, a \emph{single while loop} is a program of the following form:
\begin{equation}\label{eq:swl}
\textbf{while}~\phi~\textbf{do}~Q~\textbf{od}
\end{equation}
where $\phi$ is the loop guard from $\langle \mathit{bexpr}\rangle$ and $Q$ is a loop-free program with possibly assignment statements, conditional branches, sequential composition but without while loops.
Given a single while loop, we assign the program counter ${\verb"in"}$ to the entry point of the while loop and
the program counter ${\verb"out"}$ to the terminating point of the loop.
Below we give an example of a single while loop.

\begin{example}\label{ex:runningexample}
Consider the following single while loop:
\lstset{language=prog}
\begin{lstlisting}[mathescape]
$\verb"in"$ : while $x\geq 1$ do
       $x:=x+r$
     od
$\verb"out"$:
\end{lstlisting}


where $x$ is a program variable and $r$ is a sampling variable that observes certain fixed distributions
(e.g., a two-point distribution such that $\probm(r={-1})=\probm(r={1})=\frac{1}{2}$).
Informally, the program performs a random increment/decrement on $x$ until its value is no greater than zero.
\end{example}

\noindent{\bf The Semantics.} Since our approaches for proving almost-sure termination work basically for single while loops (in Section~\ref{sect:algorithmicmethods} we extend to probabilistic programs without nested loops),
we present the simplified semantics for single while loops.

We first introduce the notion of valuations which specify current values for program and sampling variables.
Below we fix a single while loop $P$ in the form (\ref{eq:swl}) and let $X$ (resp. $R$) be the set of program (resp. sampling) variables appearing in $P$. The size of $X,R$ is denoted by $|X|,|R|$, respectively.
We impose arbitrary linear orders on both of $X,R$ so that $X=\{x_1,\dots,x_{|X|}\}$ and $R=\{r_1,\dots,r_{|R|}\}$.
We also require that for each sampling variable $r_i \in R$, a discrete probability distribution is given.
Intuitively, at each loop iteration of $P$, the value of $r_i$ is independently sampled w.r.t the distribution.

{\em Valuations.} A \emph{program valuation} is a (column) vector $\pv\in\Zset^{|X|}$.
Intuitively, a valuation $\pv$ specifies that for each $x_i\in X$,
the value assigned is the $i$-th coordinate $\pv[i]$ of $\pv$.
Likewise, a \emph{sampling valuation} is a (column) vector $\rv\in\Zset^{|R|}$. A \emph{sampling function} $\Upsilon$ is a function assigning to every sampling variable $r\in R$ a discrete probability distribution over $\Zset$.
The discrete probability distribution $\bar{\Upsilon}$ over $\Zset^{|R|}$ is defined by:
$\bar{\Upsilon}(\rv):=\prod^{|R|}_{i=1}\Upsilon(r_i)(\rv[i])$~.



For each program valuation $\pv$, we say that $\pv$ \emph{satisfies} the loop guard $\phi$, denoted by $\pv\models\phi$ , if the formula $\phi$ holds when every appearance of a program variable is replaced by its corresponding value in $\pv$. Moreover, the loop body $Q$ in $P$ encodes a function $F: \Zset^{|X|}\times \Zset^{|R|}\rightarrow  \Zset^{|X|}$ which transforms the program valuation $\pv$ before the execution of $Q$ and the independently-sampled values in $\rv$ into the program valuation $F(\pv, \rv)$ after the execution of $Q$.

\emph{Semantics of single while loops.}
Now we present the semantics of single while loops.
Informally, the semantics is defined by a Markov chain $\mathcal{M}=(S, \textbf{P})$, where the state space $S:=\{{\verb"in"}, {\verb"out"}\}\times \Zset^{|X|}$ is a set of pairs of location and sampled values
and the probability transition function $\textbf{P}: S\times S \rightarrow [0,1] $ will be clarified later.
We call states in $S$ \emph{configurations}. A \emph{path} under the Markov chain is an infinite sequence $\{(\ell_n, \pv_n)\}_{n\ge 0}$  of configurations.
The intuition is that in a path, each $\pv_n$ (resp. $\ell_n$) is the current program valuation
(the current program counter to be executed) right before the $n$-th execution step of $P$.
Then given an initial configuration $({\verb"in"}, \pv_0)$, the probability space for $P$ is constructed as the standard one for its Markov chain over paths (for details see~\cite[Chatper 10]{BaierBook}).
We shall denote by $\probm$ the probability measure (over the $\sigma$-algebra of subsets of paths)
in the probability space for $P$ (from some fixed initial program valuation $\pv_0$).

Consider any initial program valuation $\pv$. The execution of the single while loop $P$ from  $\pv$ results in a path $\{(\ell_n,\pv_n)\}_{n\in\Nset_0}$ as follows.
Initially, $\pv_0=\pv$ and $\ell_0=\verb"in"$.
Then at each step $n$, the following two operations are performed.
First, a sampling valuation $\rv_{n}$ is obtained through samplings for all sampling variables, where the value for each sampling variable observes a predefined discrete probability distribution for the variable.
Second, we clarify three cases below:
\begin{compactitem}
\item if $\ell_{n}=\verb"in"$ and $\pv_n\models\phi$, then the program enters the loop and we have $\ell_{n+1}:=\verb"in"$, $\pv_{n+1}:=F(\pv_n,\rv_n)$, and thus we simplify the executions of $Q$ as a single computation step;
\item if $\ell_{n}=\verb"in"$ and $\pv_n\not\models\phi$, then the program enters the terminating program counter $\verb"out"$ and we have $\ell_{n+1}:=\verb"out"$, $\pv_{n+1}:=\pv_n$;
\item if $\ell_{n}=\verb"out"$ then the program stays at the program counter $\verb"out"$ and we have $\ell_{n+1}:=\verb"out"$, $\pv_{n+1}:=\pv_n$.
\end{compactitem}


Based on the informal description, we now formally define the probability transition function  $\textbf{P}$:
\begin{compactitem}
\item $\textbf{P}((\verb"in",\pv),(\verb"in",\pv'))= \sum_{\rv\in\{\rv|\pv'=F(\pv,\rv)\}}\bar{\Upsilon}(\rv)$, for any $\pv,\pv'$ such that  $\pv\models\phi$;
\item $\textbf{P}((\verb"in",\pv),(\verb"out",\pv)) =1$ for any $\pv$ such that $\pv\not\models\phi$;
\item $\textbf{P}((\verb"out",\pv),(\verb"out",\pv))=1$ for any $\pv$;
\item $\textbf{P}((\ell,\pv),(\ell',\pv'))=0$ for all other cases.
\end{compactitem}
We note that the semantics for general probabilistic programs can be defined in the same principle as
for single while loops with the help of transition structures or control-flow graphs (see ~\cite{DBLP:conf/popl/ChatterjeeFNH16,DBLP:conf/cav/ChatterjeeFG16}).
%

\noindent{\bf Almost-Sure Termination.}
In the following, we define the notion of \emph{almost-sure termination} over single while loops.
Consider a single while loop $P$. The \emph{termination-time random variable} $T$ is defined such that for any path $\{(\ell_n,\pv_n)\}_{n\in\Nset_0}$, the value of $T$ at the path is $\min\{n \mid \ell_n=\verb"out"\}$, where $\min\emptyset:=\infty$.
Then $P$ is said to be \emph{almost-surely terminating} (from some prescribed initial program valuation $\pv_0$) if $\probm(T<\infty)=1$.
Besides, we also consider bounds on tail probabilities $\probm(T\ge k)$ of non-termination within $k$ loop-iterations.
Tail bounds are important quantitative aspects that characterizes how fast the program terminates.



\section{Supermartingale Based Approach}\label{sect:supermartingale}

In this section, we present our supermartingale-based approach for proving almost-sure termination of single while loops.
We first establish new mathematical results on supermartingales, then we show how to apply these results to obtain a sound approach for proving almost-sure termination.

The following proposition is our first new mathematical result.

\begin{proposition}[Difference-bounded Supermartingales]\label{thm:supm}
Consider any difference-bounded supermartingale $\Gamma=\{X_n\}_{n\in\Nset_0}$ adapted to a filtration $\{\mathcal{F}_n\}_{n\in\Nset_0}$ satisfying the following conditions:
\begin{compactenum}
\item $X_0$ is a constant random variable;
\item for all $n\in\Nset_0$, it holds for all $\omega$ that (i) $X_n(\omega)\ge 0$ and (ii) $X_n(\omega)=0$ implies  $X_{n+1}(\omega)=0$;
\item {\em Lower Bound on Conditional Absolute Difference (LBCAD).} there exists $\delta\in(0,\infty)$ such that for all $n\in\Nset_0$, it holds a.s. that $X_n>0$ implies $\condexpv{|X_{n+1}-X_n|}{\mathcal{F}_n}\ge \delta$.
\end{compactenum}
Then $\probm(\stopping{\Gamma}<\infty)=1$ and the function $k\mapsto\probm\left(\stopping{\Gamma}\ge k\right)\in \mathcal{O}\left(\frac{1}{\sqrt{k}}\right)$.
\end{proposition}
Informally, the LBCAD condition requires that the stochastic process should have a minimal amount of vibrations at each step. The amount $\delta$ is the least amount that the stochastic process should change on its value in the next step (eg, $X_{n+1}=X_n$ is not allowed). Then it is intuitively true that if the stochastic process does not increase in expectation (i.e., a supermartingale) and satisfies the LBCAD condition, then we have at some point the stochastic processes will drop below zero.
The formal proof ideas are as follows.

\noindent{\em Key Proof Ideas.} The main idea is a thorough analysis of the
martingale
\[
Y_n:=\frac{e^{-t\cdot X_n}}{\prod_{j=0}^{n-1} \condexpv{e^{-t\cdot \left(X_{j+1}-X_{j}\right)}}{\mathcal{F}_j}} ~~(n\in\Nset_0)
\]
for some sufficiently small $t>0$ and its limit through Optional Stopping Theorem (cf. Theorem~\ref{thm:optstopping} in the appendix). We first prove that $\{Y_n\}$ is indeed a martingale.
The difference-boundedness ensures that the martingale $Y_n$ is well-defined.
Then by letting $Y_\infty:=\lim\limits_{n\rightarrow\infty} Y_{\min\{n,\stopping{\Gamma}\}}$, we prove that $\expv\left(Y_\infty\right)=\expv\left(Y_0\right)=e^{-t\cdot \expv(X_0)}$ through Optional Stopping Theorem and the LBCAD condition\enskip.
Third, we prove from basic definitions and the LBCAD condition that
\[
\expv\left(Y_\infty\right)=e^{-t\cdot\expv(X_0)}\le 1-\left(1-\left(1+\frac{\delta^2}{4}\cdot t^2\right)^{-k}\right)\cdot \probm\left( \stopping{\Gamma}\ge k\right)\enskip.
\]
By setting $t:=\frac{1}{\sqrt{k}}$ for sufficiently large $k$, one has that
\[
\probm\left(\stopping{\Gamma}\ge k\right)\le \frac{1-e^{-\frac{\expv(X_0)}{\sqrt{k}}}}{1-\left(1+\frac{\delta^2}{4}\cdot \frac{1}{k}\right)^{-k}}\enskip.
\]
It follows that $k\mapsto\probm\left(\stopping{\Gamma}\ge k\right)\in \mathcal{O}\left(\frac{1}{\sqrt{k}}\right)$.\qed

\noindent{\em Optimality of Proposition 2.} We now present two examples to illustrate two aspects of optimality of Proposition~\ref{thm:supm}. First, in Example~\ref{ex:special:supmoptimal} we show
an application on the classical symmetric random walk that the tail bound $\mathcal{O}(\frac{1}{\sqrt{k}})$ of Proposition~\ref{thm:supm} is optimal. Then in Example~\ref{ex:special:nonnegativity} we establish that the always non-negativity condition required in the second item of Proposition~\ref{thm:supm} is critical (i.e., the result does not hold without the condition).


\begin{example}\label{ex:special:supmoptimal}
Consider the family $\{Y_n\}_{n\in\Nset_0}$ of independent random variables defined as follows:
$Y_0:=1$ and each $Y_n$ ($n\ge 1$) satisfies that $\probm\left(Y_n=1\right)=\frac{1}{2}$ and  $\probm\left(Y_n=-1\right)=\frac{1}{2}$.
Let the stochastic process $\Gamma=\{X_n\}_{n\in\Nset_0}$ be inductively defined by: $X_0:=Y_0$.
$X_n$ is difference bounded since $Y_n$ is bounded.
For all $n\in\Nset_0$ we have
$X_{n+1}:=\mathbf{1}_{X_n>0}\cdot\left(X_n+Y_{n+1}\right)$.
Choose the filtration $\{\mathcal{F}_n\}_{n\in\Nset_0}$ such that every $\mathcal{F}_n$ is the smallest $\sigma$-algebra that makes $Y_0,\dots, Y_n$ measurable.
Then $\Gamma$ models the classical symmetric random walk and $X_n>0$ implies $\condexpv{|X_{n+1}-X_n|}{\mathcal{F}_n}=1$ a.s.
Thus, $\Gamma$ ensures the LBCAD condition.
From Proposition~\ref{thm:supm}, we obtain that $\probm(\stopping{\Gamma}<\infty)=1$ and $k\mapsto\probm\left(\stopping{\Gamma}\ge k\right)\in \mathcal{O}\left(\frac{1}{\sqrt{k}}\right)$.
It follows from~\cite[Theorem 4.1]{DBLP:journals/jcss/BrazdilKKV15}
that $k\mapsto\probm\left(\stopping{\Gamma}\ge k\right)\in \Omega\left(\frac{1}{\sqrt{k}}\right)$.
Hence,the tail bound $\mathcal{O}\left(\frac{1}{\sqrt{k}}\right)$ in Proposition~\ref{thm:supm} is optimal.\qed
\end{example}

\begin{example}\label{ex:special:nonnegativity}
In Proposition~\ref{thm:supm}, the condition that $X_n\ge 0$ is necessary; in other words,  it is necessary to have  $X_{\stopping{\Gamma}}=0$ rather than $X_{\stopping{\Gamma}}\le 0$ when $\stopping{\Gamma}<\infty$.
This can be observed as follows.
Consider the discrete-time stochastic processes $\{X_n\}_{n\in\Nset_0}$ and $\Gamma=\{Y_n\}_{n\in\Nset_0}$ given as follows:
\begin{compactitem}
\item the random variables $X_0,\dots,X_n,\dots$ are independent, $X_0$ is the random variable with constant value $\frac{1}{2}$ and each $X_n$ ($n\ge 1$) satisfies that $\probm\left(X_n=1\right)=e^{-\frac{1}{n^2}}$ and $\probm\left(X_n=-4\cdot n^2\right)=1-e^{-\frac{1}{n^2}}$;
\item $Y_n:=\sum_{j=0}^{n}X_j$ for $n\ge 0$.
\end{compactitem}
Let  $\mathcal{F}_n$ be the filtration which is the smallest $\sigma$-algebra that makes  $X_0,\dots,X_n$ measurable for every $n$.
Then one can show that $\Gamma$ (adapted to $\{\mathcal{F}_n\}_{n\in\Nset_0}$) satisfies integrability and the LBCAD condition,
but
$\probm\left(\stopping{\Gamma}=\infty\right)=e^{-\frac{\pi^2}{6}}>0$\enskip.
Detailed justifications are available in Appendix~\ref{app:martingale}.\qed
\end{example}

In the following, we illustrate how one can apply Proposition~\ref{thm:supm} to prove almost-sure termination of single while loops. Below we fix a single while loop $P$ in the form~(\ref{eq:swl}).
We first introduce the notion of \emph{supermartingale maps} which are a special class of functions over configurations that subjects to supermartingale-like constraints.

\begin{definition}[Supermartingale Maps]\label{def:supmmap}
A \emph{(difference-bounded) supermartingale map} (for $P$) is a function $h:\{\verb"in", \verb"out"\}\times \Zset^{|X|}\rightarrow\Rset$ satisfying that there exist real numbers $\delta,\zeta >0$ such that for all configurations $(\ell, \pv)$, the following conditions hold:
\begin{compactitem}
\item[(D1)] if $\ell=\verb"out"$ then  $h(\ell,\pv)=0$;
\item [(D2)]if $\ell=\verb"in"$ and $\pv\models\phi$, then
          (i) $h(\ell,\pv)\geq \delta$ and (ii) $h(\ell,F(\pv,\rv))\geq \delta$ for all $\rv\in \supp{\bar{\Upsilon}}$;
\item[(D3)] if $\ell= \verb"in"$ and $\pv\models\phi$ then
\begin{compactitem}
    \item[(D3.1)] $\Sigma_{\rv\in \Zset^{|R|}} \bar{\Upsilon}(\rv)\cdot h(\ell,F(\pv,\rv))\leq h (\ell,\pv) $, and
    \item[(D3.2)] $\Sigma_{\rv\in \Zset^{|R|}}\bar{\Upsilon}(\rv)\cdot |g(\ell,\pv,\rv)|\geq \delta  $   where $g(\ell,\pv,\rv):=h(\ell,F(\pv,\rv))-h (\ell,\pv) $;
    \end{compactitem}
    \item[(D4)] (for difference-boundedness) $|g(\verb"in",\pv,\rv)|\leq \zeta$ for all $\rv\in \supp{\bar{\Upsilon}}$ and $\pv\in \Zset^{|X|}$ such that $\pv\models\phi$, and $h(\verb"in", F(\pv,\rv))\le \zeta$ for all $\pv\in \Zset^{|X|}$ and $\rv\in\supp{\bar{\Upsilon}}$ such that $\pv\models\phi$ and $F(\pv,\rv)\not\models\phi$.
\end{compactitem} 
Thus, $h$ is a supermartingale map if conditions (D1)--(D3) hold. Furthermore, $h$ is difference bounded if in extra (D4) holds.
\end{definition}
Intuitively, the conditions (D1),(D2) together ensure non-negativity for the function $h$. Moreover, the difference between ``$=0$'' in (D1) and ``$\ge\delta$'' in (D2) ensures that $h$ is positive iff the program still executes in the loop.
The condition (D3.1) ensures the supermartingale condition for $h$ that the next expected value does not increase,
while the condition (D3.2) says that the expected value of the absolute change between the current and the next step is at least $\delta$, relating to the same amount in the LBCAD condition.
Finally, the condition (D4) corresponds to the difference-boundedness in supermartingales in the sense that it requires the change of value both after the loop iteration and right before the termination of the loop should be bounded by the upper bound $\zeta$.

Now we state the main theorem of this section which says that the existence of a difference-bounded supermartingale map implies almost-sure termination.

\begin{theorem}[Soundness]\label{thm:soundness}
If there exists a difference-bounded supermartingale map $h$ for $P$, then for any initial valuation $\pv_0$
we have $\probm(T< \infty)=1$ and $k\mapsto\probm(T\geq k)\in \mathcal{O}\left(\frac{1}{\sqrt{k}}\right)$.
\end{theorem}

\noindent{\em Key Proof Ideas.}
  Let $h$ be any  difference-bounded supermartingale map $h$ for the single while loop program $P$, $\pv$ be any initial valuation and $\delta, \zeta$ be the parameters in Definition~\ref{def:supmmap}. We define the stochastic process $\Gamma=\{X_n\}_{n\in\Nset_0}$ adapted to $\{\mathcal{F}_n\}_{n\in\Nset_0}$  by
  $X_n=h(\ell_n,\pv_n)$ where $\ell_n$ (resp. $\pv_n$) refers to the random variable (resp. the vector of random variables) for the program counter (resp. program valuation) at the $n$th step. Then $P$ terminates iff $\Gamma$ stops. We prove that $\Gamma$ satisfies the conditions in Proposition~\ref{thm:supm}, so that $P$ is almost-surely terminating with the same tail bound.

Theorem~\ref{thm:soundness} suggests that to prove almost-sure termination, one only needs to find a difference-bounded supermartingale map.

\begin{remark}
Informally, Theorem~\ref{thm:soundness} can be used to prove almost-sure termination of while loops where there exists a distance function (as a supermartingale map) that measures the distance of the loop to termination, for which the distance does not increase in expectation and is changed by a minimal amount in each loop iteration. The key idea to apply Theorem~\ref{thm:soundness} is to construct such a distance function.
\end{remark}
Below we illustrate an example.

\begin{example}\label{ex:dbsupmmap}
  Consider the single while loop in Example~\ref{ex:runningexample} where the distribution for $r$ is given as $\probm(r=1)=\probm(r=-1)=\frac{1}{2}$ and this program can be viewed as non-biased random walks. The program has infinite expected termination so previous approach based on ranking supermartingales cannot apply. Below we prove the almost-sure termination of the program.
  We define the difference-bounded supermartingale map $h$ by: $h(\verb"in",x)=x+1$ and $h(\verb"out",x)=0$ for every $x$.  Let $\zeta=\delta=1$. Then for every $x$, we have that
  \begin{compactitem}
  \item the condition (D1) is valid by the definition of $h$;
  \item if $\ell=\verb"in"$ and $x\geq 1$, then $h(\ell,x)=x+1\geq \delta$ and $h(\verb"in",F(x,u))=F(x,u)+1\geq x-1+1\geq \delta$ for all $u\in\supp{\bar{\Upsilon}}$. Then the condition (D2) is valid;
  \item  if $\ell=\verb"in"$ and $x\geq 1$, then $\Sigma_{u\in \Zset} \bar{\Upsilon}(u)\cdot h(\verb"in",F(x,u))
   =\frac{1}{2}((x+2)+x)\leq x+1= h (\verb"in",x) $ and
   $\Sigma_{u\in \Zset}\bar{\Upsilon}(u)\cdot |g(\verb"in",x,u)|
   =\frac{1}{2}(1+1)\geq \delta  $. Thus, we have that the condition (D3) is valid.
   \item The condition (D4) is clear as the difference is less than $1=\zeta$.
   \end{compactitem}
   It follows that $h$ is a difference-bounded supermartingale map.
   Then by Theorem~\ref{thm:soundness} it holds that the program terminates almost-surely under any initial value with tail probabilities bounded by reciprocal of square root of the thresholds.
   By similar arguments, we can show that the results still hold when we consider that the distribution of $r$ in general has bounded range, non-positive mean value and non-zero variance by letting $h(\verb"in", x)=x+K$ for some sufficiently large constant $K$.\qed
\end{example}

Now we extend Proposition~\ref{thm:supm} to general supermartingales. The extension lifts the difference-boundedness condition but derives with a weaker tail bound.

\begin{proposition}[General Supermartingales]\label{thm:supmextended}
Consider any supermartingale $\Gamma=\{X_n\}_{n\in\Nset_0}$ adapted to a filtration $\{\mathcal{F}_n\}_{n\in\Nset_0}$ satisfying the following conditions:
\begin{compactenum}
\item $X_0$ is a constant random variable;
\item for all $n\in\Nset_0$, it holds for all $\omega$ that (i) $X_n(\omega)\ge 0$ and (ii) $X_n(\omega)=0$ implies  $X_{n+1}(\omega)=0$;
\item {\em (LBCAD).} there exists $\delta\in(0,\infty)$ such that for all $n\in\Nset_0$, it holds a.s. that $X_n>0$ implies $\condexpv{|X_{n+1}-X_n|}{\mathcal{F}_n}\ge \delta$.
\end{compactenum}
Then $\probm(\stopping{\Gamma}<\infty)=1$ and the function $k\mapsto\probm\left(\stopping{\Gamma}\ge k\right)\in \mathcal{O}\left(k^{-\frac{1}{6}}\right)$.
\end{proposition}

\noindent{\em Key Proof Ideas.} The key idea is to extend the proof of Proposition~\ref{thm:supm} with the stopping times $R_M$'s ($M\in (\expv(X_0),\infty)$) defined by $R_M(\omega):=  \min\{n\mid X_{n}(\omega)\le 0\mbox{ or } X_{n}(\omega)\ge M\}$\enskip. For any $M>0$, we first define a new stochastic process $\{X'_n\}_n$ by $X'_n=\min\{X_n, M\}$ for all $n\in\Nset_0$\enskip. Then we define the discrete-time stochastic process $\{Y_n\}_{n\in\Nset_0}$ by
\[
Y_n:=\frac{e^{-t\cdot X'_n}}{\prod_{j=0}^{n-1} \condexpv{e^{-t\cdot \left(X'_{j+1}-X'_{j}\right)}}{\mathcal{F}_j}}
\]
for some appropriate positive real number $t$.
We prove that $\{Y_n\}_{n\in\Nset_0}$ is still a martingale.
Then from Optional Stopping Theorem, by letting $Y_\infty:=\lim\limits_{n\rightarrow\infty} Y_{\min\{n, R_M\}}$,
we also have $\expv\left(Y_\infty\right)=\expv\left(Y_0\right)=e^{-t\cdot \expv(X_0)}$\enskip.
Thus, we can also obtain similarly that
\[
\expv\left(Y_\infty\right)=e^{-t\cdot\expv(X_0)}\le 1-\left(1-\left(1+\frac{\delta^2}{16}\cdot t^2\right)^{-k}\right)\cdot \probm\left( R_M\ge k\right)\enskip.
\]
For $k\in \Theta(M^6)$ and $t=\frac{1}{\sqrt{k}}$, we obtain $\probm\left(R_M\ge k\right)\in \mathcal{O}(\frac{1}{\sqrt{k}})$\enskip. Hence, $\probm\left(R_M=\infty\right)=0$.
By Optional Stopping Theorem, we have $\expv(X_{R_M})\le \expv(X_0)$. Furthermore, we have by Markov's Inequality that
$\probm(X_{R_M}\ge M)\le \frac{\expv(X_{R_M})}{M}\le \frac{\expv(X_0)}{M}$\enskip.
Thus, for sufficiently large $k$ with $M\in \Theta(k^{\frac{1}{6}})$, we can deduce that
$\probm(\stopping{\Gamma}\ge k) \le \probm(R_M\ge k) + \probm(X_{R_M}\ge M)\in \mathcal{O}(\frac{1}{\sqrt{k}}+\frac{1}{\sqrt[6]{k}})$.\qed

\begin{remark}
Similar to Theorem~\ref{thm:soundness}, we can establish a soundness result for general supermartingales. The result simply says that the existence of a (not necessarily difference-bounded) supermartingale map implies almost-sure termination and a weaker tail bound $\mathcal{O}(k^{-\frac{1}{6}})$.
\end{remark}

The following example illustrates the application of Proposition~\ref{thm:supmextended} on a single while loop with unbounded difference.


\begin{example}\label{ex:dbsupmmapextended}
Consider the following single while loop program

\lstset{language=prog}
\begin{lstlisting}[mathescape]
$\verb"in"$ : while $x\geq 1$ do
       $x:=x+r\cdot\lfloor\sqrt{x}\rfloor$
     od
$\verb"out"$:
\end{lstlisting}

where the distribution for $r$ is given as $\probm(r=1)=\probm(r=-1)=\frac{1}{2}$.
The supermartingale map $h$ is defined as the one in Example~\ref{ex:dbsupmmap}.
   In this program, $h$ is not difference-bounded as $\lfloor\sqrt{x}\rfloor$ is not bounded. Thus, $h$ satisfies the conditions except (D4) in Definition~\ref{def:supmmap}.
   We now construct a stochastic process $\Gamma=\{X_n=h(\ell_n,\pv_n)\}_{n\in\Nset_0}$ which meets the requirements of Proposition~\ref{thm:supmextended}.
   It follows that the program terminates almost-surely under any initial value with tail probabilities bounded by $\mathcal{O}\left(k^{-\frac{1}{6}}\right)$.
   In general, if $r$ observes a distribution with bounded range $[-M, M]$, non-positive mean and non-zero variance, then we can still prove the same result as follows. We choose a sufficiently large constant $K\geq\frac{M^2}{4}+1$ so that the function $h$ with $h(\verb"in", x)=x+K$ is still a supermartingale map since the non-negativity of  $h(\verb"in",x)= x-M\cdot\sqrt{x}+K=(\sqrt{x}-\frac{M}{2})^2-\frac{M^2}{4}+K\ge -\frac{M^2}{4}+K$ for all $x\ge 0$. \qed
\end{example}


\section{Central Limit Theorem Based Approach}\label{sect:central}

We have seen in the previous section a supermartingale-based approach for proving almost-sure termination.
However by Example~\ref{ex:special:nonnegativity}, an inherent restriction is that the supermartingale should be non-negative. In this section, we propose a new approach through Central Limit Theorem that can drop this requirement but requires in extra an independence condition.

We first state the well-known Central Limit Theorem~\cite[Chapter 18]{probabilitycambridge}.

\begin{theorem}[Lindeberg-L\'{e}vy's Central Limit Theorem]\label{thm:lindclt}
  Suppose $\{X_1,X_2,\ldots\}$ is a sequence of independent and identically distributed random variables with $\expv(X_i)=\mu$ and $\var(X_i)=\sigma^2>0$ is finite.
  Then as $n$ approaches infinity, the random variables
$\sqrt{n}((\frac{1}{n}\sum_{i=1}^{n}X_i)-\mu)$ converge in distribution to a normal $(0,\sigma^2)$.
  In the case $\sigma>0$, we have for every real number $z$
  \[\lim_{n\rightarrow\infty}\probm(\sqrt{n}((\frac{1}{n}\sum_{i=1}^{n}X_i)-\mu)\leq z)=\Phi(\frac{z}{\sigma}),\]
  where $\Phi(x)$ is the standard normal cumulative distribution functions evaluated at $x$.
\end{theorem}

The following lemma is key to our approach, proved by Central Limit Theorem.

\begin{lemma}\label{thm:clt}
  Let $\{R_n\}_{n\in\Nset}$ be a sequence of independent and identically distributed random variables with expected value $\mu=\expv(R_n)\leq 0$ and finite variance $\var(R_n)=\sigma^2>0$ for every $n\in\Nset$. For every $x\in \Rset$, let $\Gamma=\{X_n\}_{n\in\Nset_0}$ be a discrete-time stochastic process, where $X_0=x$ and $X_{n}=x+\Sigma^{n}_{k=1}R_k$ for $n\geq 1$.  Then there exists a constant $p>0$, for any $x$, we have $\probm(Z_\Gamma<\infty)\geq p$.
\end{lemma}

\begin{proof}
  According to the Central Limit Theorem (Theorem~\ref{thm:lindclt}),
  \[
  \lim_{n\rightarrow\infty}\probm(\sqrt{n}(\frac{X_n-x}{n}-\mu)\leq z)=\Phi(\frac{z}{\sigma})
  \]
   holds for every real number $z$. Note that
  \[ \probm(\sqrt{n}(\frac{X_n-x}{n}-\mu)\leq z)= \probm(X_n\leq \sqrt{n}\cdot z+n\cdot \mu+x )
  \leq \probm(X_n\leq \sqrt{n}\cdot z+x ).\]
  Choose $z=-1$. Then we have $\probm(X_n \leq 0) \ge \probm(X\leq -\sqrt{n}+x)$ when $n > x^2$. Now we fix a proper $\epsilon<\Phi(\frac{-1}{\sigma})$, and get $n_0(x)$ from the limit form equation such that for all $n>\max\{n_0(x),x^2\}$ we have
  \[\probm(X_n \leq 0)\ge \probm(X\leq -\sqrt{n} +x)\geq \probm(\sqrt{n}(\frac{X_n-X_0}{n}-\mu)\leq -1) \geq \Phi(\frac{-1}{\sigma})-\epsilon= p >0.  \]
  Since $X_n\le 0$ implies $Z_\Gamma<\infty$, we obtain that $\probm(Z_\Gamma<\infty)\geq p$ for every $x$.
\end{proof}

\noindent{\em Incremental Single While Loops.}
Due to the independence condition required by Central Limit Theorem, we need to consider special classes of single while loops. We say that a single while loop  $P$ in the form ~(\ref{eq:swl}) is \emph{incremental} if $Q$ is a sequential composition of assignment statements of the form $x:=x+\sum_{i=1}^{|R|}c_i\cdot r_i$ where $x$ is a program variable, $r_i$'s are sampling variables and $c_i$'s are constant coefficients for sampling variables.
We then consider incremental single while loops.
For incremental single while loops, the function $F$ for the loop body $Q$ is incremental, i.e., $F(\pv,\rv)=\pv+\mathbf{A}\cdot \rv$ for some constant matrix $\mathbf{A}\in{\Zset^{|X|\times |R|}}$.
\begin{remark}
By Example~\ref{ex:special:nonnegativity}, previous approaches cannot handle incremental single while loops with unbounded range of sampling variables (so that a supermartingale with a lower bound on its values may not exist). On the other hand, any additional syntax such as conditional branches or assignment statements like $x:=2\cdot x+r$ will result in an increment over certain program variables that is dependent on the previous executions of the program, breaking the independence condition.
\end{remark}

To prove almost-sure termination of incremental single while loops through Central Limit Theorem, we introduce the notion of \emph{linear progress functions}. Below we fix an incremental single while loop $P$ in the form ~(\ref{eq:swl}).

\begin{definition}[Linear Progress Functions]\label{def:progfun}
A \emph{linear progress function} for $P$ is a function $h:\Zset^{|X|}\rightarrow \Rset$ satisfying the following conditions:
\begin{compactitem}
\item[(L1)] there exists $\mathbf{a}\in\Rset^{|X|}$ and $c\in\Rset$ such that $h(\pv)=\mathbf{a}^{\mathrm{T}}\cdot \pv+c$ for all program valuations $\pv$;
\item[(L2)] for all program valuations $\pv$, if $\pv\models\phi$ then $h(\pv)>0$;
\item[(L3)] $\sum_{i=1}^{|R|}a_i\cdot \mu_i\le 0$ and $\sum_{i=1}^{|R|}a^2_i\cdot \sigma^2_i>0$, where
\begin{compactitem}
\item $(a_1,\dots,a_{|R|})=\mathbf{a}^{\mathrm{T}}\cdot\mathbf{A}$,
\item $\mu_i$ (resp. $\sigma^2_i$) is the mean (resp. variance) of the distribution $\Upsilon(r_i)$, for $1\le i\le |R|$.
\end{compactitem}
\end{compactitem}
\end{definition}
Intuitively, the condition (L1) says that the function should be linear; the condition (L2) specifies that if the value of $h$ is non-positive, then the program terminates; the condition (L3) enforces that the mean of $\mathbf{a}^{\mathrm{T}}\cdot\mathbf{A}\cdot \mathbf{u}$ should be non-positive, while its variance should be non-zero. The main theorem of this section is then as follows.
\begin{theorem}[Soundness]\label{thm:cltsound}
   For any incremental single while loop program $P$,  if there exists a linear progress function for $P$, then for any initial valuation $\pv_0$
we have $\probm(T< \infty)=1$.
\end{theorem}

\begin{proof}
  Let $h(\pv)=\mathbf{a}^{\mathrm{T}}\cdot \pv+c$ be a linear progress function for $P$. We define the stochastic process $\Gamma=\{X_n\}_{n\in \Nset_0}$ by $X_n=h(\pv_n)$, where $\pv_n$ is the vector of random variables that represents the program valuation at the $n$th execution step of $P$.  Define $R_n:=X_n-X_{n-1}$. We have $R_n=X_n-X_{n-1}=h(\pv_n)-h(\pv_{n-1})=h(\pv_{n-1}+\textbf{A}\cdot\rv_{n})-h(\pv_{n-1})=\mathbf{a}^{\mathrm{T}}\cdot \textbf{A}\cdot\rv_{n}$ for $n\geq 1$.  Thus, $\{R_n\}_{n\in\Nset}$ is a sequence of independent and identically distributed random variables. We have
  $\mu:=\expv(R_n)\leq0$ and $\sigma^2:=\mathbf{Var}(R_n)>0$
  by the independency of $r_i$'s and the condition (L3) in Definition~\ref{def:progfun}. Now we can apply Lemma~\ref{thm:clt} and obtain that there exists a constant $p>0$ such that for any initial program valuation $\pv_0$, we have $\probm(Z_\Gamma<\infty)\geq p$.
   By the recurrence property of Markov chain, we have $\{X_n\}$ is almost-surely stopping. Notice that from (L2), $0\geq X_n=h(\pv_n)$ implies $\pv_n \not\models\phi$ and (in the next step) termination of the single while loop. Hence,we have that $P$ is almost-surely terminating under any initial program valuation $\pv_0$.\qed
\end{proof}

Theorem~\ref{thm:cltsound} can be applied to prove almost-sure termination of while loops whose increments are independent, but the value change in one iteration is not bounded. Thus, Theorem~\ref{thm:cltsound} can handle programs which Theorem~\ref{thm:soundness} and Proposition~\ref{thm:supmextended} as well as previous supermartingale-based methods cannot.

In the following, we present several examples, showing that Theorem~\ref{thm:cltsound} can handle sampling variables with unbounded range which previous approaches cannot handle.


\begin{example}\label{ex:cltorigin}
Consider the program in Example~\ref{ex:runningexample} where we let $r$ be a two-sided geometric distribution sampling variable such that $\probm(r=k>0)=\frac{(1-p)^{k-1}p}{2}$ and $\probm(r=k<0)=\frac{(1-p)^{-k-1}p}{2}$ for some $0<p<1$. First note that by the approach in~\cite{DBLP:journals/pacmpl/AgrawalC018}, we can prove that this program has
infinite expected termination time, and thus previous ranking-supermartingale based approach cannot be applied.
Also note that the value that $r$ may take has no lower bound. This means that we can hardly obtain the almost-sure termination by finding a proper supermartingale map that satisfy both the non-negativity condition and the non-increasing condition.
Now we apply Theorem~\ref{thm:cltsound}.
Choose $h(x)=x$. It follows directly that both (L1) and (L2) hold. Since $\expv(r)=0$ for symmetric property and $0<\var(r)=\expv(r^2)-\expv^2(r)=\expv(r^2)=\expv(Y^2)=\var(Y)-\expv^2(Y)<\infty$ where $Y$ is the standard geometric distribution with parameter $p$, we have (L3) holds. Thus, $h$ is a legal linear
 progress function and
 this program is almost-sure terminating by Theorem~\ref{thm:cltsound}.\qed
\end{example}


\begin{example}\label{ex:clt}
Consider the following program with a more complex loop guard.
  \lstset{language=prog}
\begin{lstlisting}[mathescape]
$\verb"in"$ : while $y>x^2$ do
       $x:=x+r_1$;
       $y:=y+r_2$
     od
$\verb"out"$:
\end{lstlisting}

This program terminates when the point on the plane leaves the area above
the parabola by a two-dimensional random walk.
We suppose that $\mu_1=\expv(r_1), \mu_2=\expv(r_2)$ are both positive and $0<\var(r_1),\var(r_2)<\infty$ .
Now we are to prove the program is almost-surely terminating by constructing a linear progress function $h$. The existence of a linear progress function renders the result valid by Theorem~\ref{thm:cltsound}.
Let $h(x,y)=-\mu_2\cdot x+\mu_1\cdot y+\frac{\mu_2^2}{4\mu_1}$.
If $y>x^2$, then $h(x,y)>\mu_1\cdot x^2-\mu_2\cdot x+\frac{\mu_2^2}{4\mu_1}=\mu_1(x-\frac{\mu_2}{2\mu_1})^2\geq0$.
From $\textbf{a}^{\textrm{T}}\cdot \textbf{A}\cdot (\expv(r_1),\expv(r_2))^{\mathrm{T}}= -\mu_2\cdot\mu_1+\mu_1\cdot\mu_2 =0$,
we have $h$ is a legal linear progress function for $P$.
Thus, $P$ is almost-surely terminating.\qed
\end{example}

\section{Algorithmic Methods and Extensions}\label{sect:algorithmicmethods}

In this section, we discuss possible extensions for our results, such as algorithmic methods, real-valued program variables, non-determinism.

\noindent{\em Algorithmic Methods.} Since program termination is generally undecidable, algorithms for proving termination of programs require certain restrictions.
A typical restriction adopted in previous ranking-supermartingale-based algorithms~\cite{SriramCAV,DBLP:conf/cav/ChatterjeeFG16,DBLP:conf/popl/ChatterjeeFNH16,ChatterjeeNZ2017} is a fixed template for ranking supermartingales.
Such a template fixes a specific form for ranking supermartingales.
In general, a ranking-supermartingale-based algorithm first establishes a template with unknown coefficients for a ranking supermartingale.
The constraints over those unknown coefficients are inherited from the properties of the ranking supermartingale.
Finally, constraints are solved using either linear programming or semidefinite programming.

This algorithmic paradigm can be directly extended to our supermartingale-based approaches.
First, an algorithm can establish a linear or polynomial template with unknown coefficients for a supermartingale map.
Then our conditions from supermartingale maps (namely (D1)--(D4)) result in constraints on the unknown coefficients.
Finally, linear or semidefinite programming solvers can be applied to obtain the concrete values for those unknown coefficients.

For our CLT-based approach, the paradigm is more direct to apply.
We first establish a linear template with unknown coefficients.
Then we just need to find suitable coefficients such that (i) the difference has non-positive mean value and non-zero variance and (ii) the condition (D5) holds, which again reduces to linear programming.

In conclusion, previous algorithmic results can be easily adapted to our approaches.

\noindent{\em Real-Valued Program Variables.} A major technical difficulty to handle real numbers is the \emph{measurability} condition~(cf.~\cite[Chapter 3]{probabilitycambridge}).
For example, we need to ensure that our supermartingale map is measurable in some sense.
The measurability condition also affects our CLT-based approach as it is more difficult to prove the recurrence property in continuous-state-space case.
However, the issue of measurability is only technical and not fundamental, and thus we believe that our approaches can be extended to real-valued program variables and continuous samplings such as uniform or Gaussian distribution.

\noindent{\em Non-determinism.} In previous works, non-determinism is handled by ensuring related properties in each non-deterministic branch.
For examples, previous results on ranking supermartingales~\cite{SriramCAV,DBLP:conf/cav/ChatterjeeFG16,DBLP:conf/popl/ChatterjeeFNH16} ensures that
the conditions for ranking supermartingales should hold for all non-deterministic branches if we have demonic non-determinism, and for at least one non-deterministic branch if we have angelic non-determinism.
Algorithmic methods can then be adapted depending on whether the non-determinism is demonic or angelic.

Our supermartingale-based approaches can be easily extended to handle non-determinism. If we have demonic non-determinism in the single while loop, then we just ensure that
the supermartingale map satisfies the conditions (D1)--(D4) no matter which demonic branch is taken.
Similarly, for angelic non-determinism, we just require that the conditions (D1)--(D4) hold for at least one angelic branch. Then algorithmic methods can be developed to handle non-determinism.

On the other hand, we cannot extend our CLT-based approach directly to non-determinism. The reason is that under history-dependent schedulers, the sampled value at the $n$th step may not be independent of those in the previous step. In this sense, we cannot apply Central Limit Theorem since it requires the independence condition.
Hence,we need to develop new techniques to handle non-determinism in the cases from Section~\ref{sect:central}.
We leave this interesting direction as a future work.

%
%
%

 \section{Applicability of Our Approaches}

Up till now, we have illustrated our supermartingale based and Central-Limit-Theorem based approach only over single probabilistic while loops.
A natural question arises whether our approach can be applied to programs with more complex structures.
Below we discuss this point.

First, we demonstrate that our approaches can in principle be applied to all probabilistic programs without nested loops, as is done by a simple compositional argument.


 \begin{remark}[Compositionality] \label{rmk:comp}
 We note that the property of almost-sure termination for all initial program valuations are closed under sequential composition and conditional branches. Thus, it suffices to consider single while loops, and the results extend straightforwardly to all imperative probabilistic programs without nested loops.
It follows that our approaches can in principle handle all probabilistic programs without nested loops.
We plan the interesting direction of compositional reasoning for nested probabilistic loops as a future work.\qed
\end{remark}

Second, we show that our approaches cannot be directly extended to nested probabilistic loops. The following remark present the details.

\begin{remark}
%
Consider a probabilistic nested loop
\[
\textbf{while}~\phi~\textbf{do}~P~\textbf{od}
\]
where $P$ is another probabilistic while loop. On one hand, if we apply supermartingales directly to such programs, then either (i) the value of an appropriate supermartingale may grow unboundedly  below zero due to the possibly unbounded termination time of the loop $P$, which breaks the necessary non-negativity condition (see Example~\ref{ex:special:nonnegativity}), or (ii) we restrict supermartingales to be non-negative on purpose in the presence of nested loops, but then we can only handle simple nested loops (e.g., inner and outer loops do not interfere). On the other hand, the CLT-based approach rely on independence, and cannot be applied to nested loops since the nesting loop will make the increment of the outer loop not independent.\qed
\end{remark}

To summarize, while our approaches apply to all probabilistic programs {without} nested loops, new techniques beyond supermartingales and Central Limit Theorem are needed to handle general nested loops.


\section{Related Works}\label{sect:realtedwork}

We compare our approaches with other approaches on termination of probabilistic programs.
As far as we know, there are two main classes of approaches for proving termination of probabilistic programs,
namely (ranking) supermartingales and proof rules.


\noindent{\em Supermartingale-Based Approach.}
First, we point out the major difference between our approaches and ranking-supermartingale-based approaches~\cite{SriramCAV,BG05,HolgerPOPL,DBLP:conf/cav/ChatterjeeFG16,DBLP:conf/popl/ChatterjeeFNH16}.
The difference is that ranking-supermartingale-based approaches can only be applied to programs with finite expected termination time.
Although in \cite{DBLP:journals/pacmpl/AgrawalC018} a notion of lexicographic ranking supermartingales is proposed to
prove almost-sure termination of compositions of probabilistic while loops, the approach still relies on ranking supermartingales for a single loop, and thus cannot be applied to single while loops with infinite expected termination time.
In our paper, we target probabilistic programs with infinite expected termination time, and thus our approaches can handle programs that
ranking-supermartingale-based approaches cannot handle.

Then we remark on the most-related work~\cite{MclverMorgan2016} which also considered supermartingale-based approach for almost-sure termination.
Compared with our supermartingale-based approach, the approach in~\cite{MclverMorgan2016} relaxes the LBCAD condition in Proposition~\ref{thm:supm} so that a more general result on almost-sure termination is obtained but the tail bounds cannot be guaranteed, while our results can derive optimal tail bounds.
Moreover, the approach in ~\cite{MclverMorgan2016} requires that the values taken by the supermartingale should have a lower bound,
while our CLT-based approach do not require this restriction and hence can handle almost-sure terminating programs that cannot be handled in ~\cite{MclverMorgan2016}.
Finally, our supermartingale-based results are independent of~\cite{MclverMorgan2016} (see arXiv versions~\cite{DBLP:journals/corr/McIverM16} and ~\cite[Theorem 5 and Theorem 6]{DBLP:journals/corr/ChatterjeeF17}).


\smallskip\noindent{\em Proof-Rule-Based Approach.}
In this paper, we consider the supermartingale based approach for probabilistic
programs. An alternative approach is based on the notion of proof rules~\cite{DBLP:conf/esop/KaminskiKMO16,DBLP:conf/lics/OlmedoKKM16}.
In the approach of proof rules, a set of rules is proposes following which one can prove termination.
Currently, the approach of proof rules is also restricted to finite termination as the proof rules also require certain quantity to decrease in expectation,
similar to the requirement of ranking supermartingales.

\smallskip\noindent{\em Potential-Function-Based Approach.}
Recently, there is another approach through the notion of \emph{potential functions}~\cite{DBLP:conf/pldi/NgoC018}.
This approach is similar to ranking supermartingales, and can derive upper bounds for expected termination time and cost.
In principle, the major difference between the approaches of ranking supermartingales and potential functions lies in algorithmic details.
In the approach of (ranking) supermartingales, the unknown coefficients in a template are solved by linear/semidefinite programming, while the approach of potential functions solves the template through inference rules.


\section{Conclusion}\label{sect:conclusion}

In this paper, we studied sound approaches for proving almost-sure termination of probabilistic programs with integer-valued program variables.
We first presented new mathematical results for supermartingales which yield new sound approaches for proving almost-sure termination of simple probabilistic while loops. Based on the above results, we presented sound supermartingale-based approaches for proving almost-sure termination of simple probabilistic while loops.
Besides almost-sure termination, our supermartingale-based approach is the first to give (optimal) bounds on tail probabilities of non-termination within a given number of steps.
Then we proposed a new sound approach through Central Limit Theorem that can prove almost-sure termination of examples that no previous approaches can handle.
Finally, we have shown possible extensions of our approach to algorithmic methods, non-determinism, real-valued program variables, and demonstrated that in principle our approach can handle all probabilistic programs without nested loops through simple compositional reasoning. 


%





\bibliographystyle{plainurl}
\bibliography{PL}

\begin{thebibliography}{10}

\bibitem{DBLP:journals/pacmpl/AgrawalC018}
Sheshansh Agrawal, Krishnendu Chatterjee, and Petr Novotn{\'{y}}.
\newblock Lexicographic ranking supermartingales: an efficient approach to
  termination of probabilistic programs.
\newblock {\em {PACMPL}}, 2({POPL}):34:1--34:32, 2018.
\newblock URL: \url{http://doi.acm.org/10.1145/3158122}, \href
  {http://dx.doi.org/10.1145/3158122} {\path{doi:10.1145/3158122}}.

\bibitem{BaierBook}
Christel Baier and Joost-Pieter Katoen.
\newblock {\em Principles of model checking}.
\newblock MIT Press, 2008.

\bibitem{BG05}
Olivier Bournez and Florent Garnier.
\newblock Proving positive almost-sure termination.
\newblock In {\em RTA}, pages 323--337, 2005.

\bibitem{DBLP:conf/cav/BradleyMS05}
Aaron~R. Bradley, Zohar Manna, and Henny~B. Sipma.
\newblock Linear ranking with reachability.
\newblock In {\em CAV}, pages 491--504, 2005.

\bibitem{DBLP:journals/jcss/BrazdilKKV15}
Tom{\'{a}}s Br{\'{a}}zdil, Stefan Kiefer, Anton{\'{\i}}n Kucera, and
  Ivana~Hutarov{\'{a}} Varekov{\'{a}}.
\newblock Runtime analysis of probabilistic programs with unbounded recursion.
\newblock {\em J. Comput. Syst. Sci.}, 81(1):288--310, 2015.

\bibitem{SriramCAV}
Aleksandar Chakarov and Sriram Sankaranarayanan.
\newblock Probabilistic program analysis with martingales.
\newblock In {\em CAV}, pages 511--526, 2013.

\bibitem{DBLP:journals/corr/ChatterjeeF17}
Krishnendu Chatterjee and Hongfei Fu.
\newblock Termination of nondeterministic recursive probabilistic programs.
\newblock {\em CoRR}, abs/1701.02944, Jan. 2017.

\bibitem{DBLP:conf/cav/ChatterjeeFG16}
Krishnendu Chatterjee, Hongfei Fu, and Amir~Kafshdar Goharshady.
\newblock Termination analysis of probabilistic programs through
  {P}ositivstellensatz's.
\newblock In {\em CAV}, pages 3--22, 2016.

\bibitem{DBLP:conf/popl/ChatterjeeFNH16}
Krishnendu Chatterjee, Hongfei Fu, Petr Novotn{\'{y}}, and Rouzbeh
  Hasheminezhad.
\newblock Algorithmic analysis of qualitative and quantitative termination
  problems for affine probabilistic programs.
\newblock In {\em POPL}, pages 327--342, 2016.

\bibitem{ChatterjeeNZ2017}
Krishnendu Chatterjee, Petr Novotn\'{y}, and {\DJ}or{\dj}e \v{Z}ikeli\'{c}.
\newblock Stochastic invariants for probabilistic termination.
\newblock In {\em POPL}, pages 145--160, 2017.

\bibitem{DBLP:conf/tacas/ColonS01}
Michael Col{\'{o}}n and Henny Sipma.
\newblock Synthesis of linear ranking functions.
\newblock In {\em TACAS}, pages 67--81, 2001.

\bibitem{EGK12}
Javier Esparza, Andreas Gaiser, and Stefan Kiefer.
\newblock Proving termination of probabilistic programs using patterns.
\newblock In {\em CAV}, pages 123--138, 2012.

\bibitem{HolgerPOPL}
Luis Mar{\'{\i}}a~Ferrer Fioriti and Holger Hermanns.
\newblock Probabilistic termination: Soundness, completeness, and
  compositionality.
\newblock In {\em POPL}, pages 489--501, 2015.

\bibitem{rwfloyd1967programs}
Robert~W. Floyd.
\newblock Assigning meanings to programs.
\newblock {\em Mathematical Aspects of Computer Science}, 19:19--33, 1967.

\bibitem{Foster53}
F.~G. Foster.
\newblock On the stochastic matrices associated with certain queuing processes.
\newblock {\em The Annals of Mathematical Statistics}, 24(3):355--360, 1953.

\bibitem{DBLP:conf/icse/GordonHNR14}
Andrew~D. Gordon, Thomas~A. Henzinger, Aditya~V. Nori, and Sriram~K. Rajamani.
\newblock Probabilistic programming.
\newblock In James~D. Herbsleb and Matthew~B. Dwyer, editors, {\em {FOSE}},
  pages 167--181. {ACM}, 2014.

\bibitem{kaelbling1998planning}
L.~P. Kaelbling, M.~L. Littman, and A.~R. Cassandra.
\newblock Planning and acting in partially observable stochastic domains.
\newblock {\em Artificial intelligence}, 101(1):99--134, 1998.

\bibitem{LearningSurvey}
L.~P. Kaelbling, M.~L. Littman, and A.~W. Moore.
\newblock Reinforcement learning: A survey.
\newblock {\em JAIR}, 4:237--285, 1996.

\bibitem{DBLP:conf/mfcs/KaminskiK15}
Benjamin~Lucien Kaminski and Joost{-}Pieter Katoen.
\newblock On the hardness of almost-sure termination.
\newblock In Giuseppe~F. Italiano, Giovanni Pighizzini, and Donald Sannella,
  editors, {\em {MFCS}}, volume 9234 of {\em LNCS}, pages 307--318. Springer,
  2015.

\bibitem{DBLP:conf/esop/KaminskiKMO16}
Benjamin~Lucien Kaminski, Joost{-}Pieter Katoen, Christoph Matheja, and
  Federico Olmedo.
\newblock Weakest precondition reasoning for expected run-times of
  probabilistic programs.
\newblock In Peter Thiemann, editor, {\em {ESOP}}, volume 9632 of {\em LNCS},
  pages 364--389. Springer, 2016.

\bibitem{Kemeny}
J.G. Kemeny, J.L. Snell, and A.W. Knapp.
\newblock {\em Denumerable {Markov} Chains}.
\newblock D. Van Nostrand Company, 1966.

\bibitem{prism}
Marta~Z. Kwiatkowska, Gethin Norman, and David Parker.
\newblock Prism 4.0: Verification of probabilistic real-time systems.
\newblock In {\em CAV}, LNCS 6806, pages 585--591, 2011.

\bibitem{MM04}
Annabelle McIver and Carroll Morgan.
\newblock Developing and reasoning about probabilistic programs in \emph{pGCL}.
\newblock In {\em PSSE}, pages 123--155, 2004.

\bibitem{MM05}
Annabelle McIver and Carroll Morgan.
\newblock {\em Abstraction, Refinement and Proof for Probabilistic Systems}.
\newblock Monographs in Computer Science. Springer, 2005.

\bibitem{DBLP:journals/corr/McIverM16}
Annabelle McIver and Carroll Morgan.
\newblock A new rule for almost-certain termination of probabilistic and
  demonic programs.
\newblock {\em CoRR}, abs/1612.01091, Dec. 2016.

\bibitem{MclverMorgan2016}
Annabelle McIver, Carroll Morgan, Benjamin~Lucien Kaminski, and Joost{-}Pieter
  Katoen.
\newblock A new proof rule for almost-sure termination.
\newblock {\em {PACMPL}}, 2({POPL}):33:1--33:28, 2018.
\newblock URL: \url{http://doi.acm.org/10.1145/3158121}, \href
  {http://dx.doi.org/10.1145/3158121} {\path{doi:10.1145/3158121}}.

\bibitem{DBLP:conf/pldi/NgoC018}
Van~Chan Ngo, Quentin Carbonneaux, and Jan Hoffmann.
\newblock Bounded expectations: resource analysis for probabilistic programs.
\newblock In Jeffrey~S. Foster and Dan Grossman, editors, {\em Proceedings of
  the 39th {ACM} {SIGPLAN} Conference on Programming Language Design and
  Implementation, {PLDI} 2018, Philadelphia, PA, USA, June 18-22, 2018}, pages
  496--512. {ACM}, 2018.
\newblock URL: \url{http://doi.acm.org/10.1145/3192366.3192394}, \href
  {http://dx.doi.org/10.1145/3192366.3192394}
  {\path{doi:10.1145/3192366.3192394}}.

\bibitem{DBLP:conf/lics/OlmedoKKM16}
Federico Olmedo, Benjamin~Lucien Kaminski, Joost-Pieter Katoen, and Christoph
  Matheja.
\newblock Reasoning about recursive probabilistic programs.
\newblock In {\em LICS}, pages 672--681, 2016.

\bibitem{PazBook}
A.~Paz.
\newblock {\em Introduction to probabilistic automata (Computer science and
  applied mathematics)}.
\newblock Academic Press, 1971.

\bibitem{DBLP:conf/vmcai/PodelskiR04}
Andreas Podelski and Andrey Rybalchenko.
\newblock A complete method for the synthesis of linear ranking functions.
\newblock In {\em VMCAI}, pages 239--251, 2004.

\bibitem{Rabin63}
M.O. Rabin.
\newblock Probabilistic automata.
\newblock {\em Information and Control}, 6:230--245, 1963.

\bibitem{SumitPLDI}
Sriram Sankaranarayanan, Aleksandar Chakarov, and Sumit Gulwani.
\newblock Static analysis for probabilistic programs: inferring whole program
  properties from finitely many paths.
\newblock In {\em PLDI}, pages 447--458, 2013.

\bibitem{DBLP:conf/pods/SohnG91}
Kirack Sohn and Allen~Van Gelder.
\newblock Termination detection in logic programs using argument sizes.
\newblock In {\em PODS}, pages 216--226, 1991.

\bibitem{DBLP:conf/aistats/MeentYMW15}
Jan{-}Willem van~de Meent, Hongseok Yang, Vikash Mansinghka, and Frank Wood.
\newblock Particle gibbs with ancestor sampling for probabilistic programs.
\newblock In {\em {AISTATS}}, 2015.

\bibitem{probabilitycambridge}
David Williams.
\newblock {\em {P}robability with {M}artingales}.
\newblock Cambridge University Press, 1991.

\bibitem{DBLP:conf/concur/Yang17}
Hongseok Yang.
\newblock Probabilistic programming (invited talk).
\newblock In Roland Meyer and Uwe Nestmann, editors, {\em {CONCUR}}, volume~85
  of {\em LIPIcs}, pages 3:1--3:1. Schloss Dagstuhl - Leibniz-Zentrum fuer
  Informatik, 2017.

\end{thebibliography}

\clearpage

\appendix

\section{Properties for Conditional Expectation}\label{app:condexpv}

Conditional expectation has the following properties for any random variables $X,Y$ and $\{X_n\}_{n\in\Nset_0}$ (from a same probability space) satisfying $\expv(|X|)<\infty,\expv(|Y|)<\infty, \expv(|X_n|)<\infty$ ($n\ge 0$) and any suitable sub-$\sigma$-algebras $\mathcal{G},\mathcal{H}$:
\begin{compactitem}
\item[(E4)] $\expv\left(\condexpv{X}{\mathcal{G}}\right)=\expv(X)$ ;
\item[(E5)] if $X$ is $\mathcal{G}$-measurable, then $\condexpv{X}{\mathcal{G}}=X$ a.s.;
\item[(E6)] for any real constants $b,d$,
\[
\condexpv{b\cdot X+d\cdot Y}{G}=b\cdot\condexpv{X}{G}+d\cdot \condexpv{Y}{G}\mbox{ a.s.;}
\]
\item[(E7)] if $\mathcal{H}\subseteq\mathcal{G}$, then $\condexpv{\condexpv{X}{\mathcal{G}}}{\mathcal{H}}=\condexpv{X}{\mathcal{H}}$ a.s.;
\item[(E8)] if $Y$ is $\mathcal{G}$-measurable and $\expv(|Y|)<\infty$, $\expv(|Y\cdot X|)<\infty$,
then
\[
\condexpv{Y\cdot X}{\mathcal{G}}=Y\cdot\condexpv{X}{\mathcal{G}}\mbox{ a.s.;}
\]
\item[(E9)] if $X$ is independent of $\mathcal{H}$, then $\condexpv{X}{\mathcal{H}}=\expv(X)$ a.s., where $\expv(X)$ here is deemed as the random variable with constant value $\expv(X)$;
\item[(E10)] if it holds a.s that $X\ge 0$, then $\condexpv{X}{\mathcal{G}}\ge 0$ a.s.;
\item[(E11)] if it holds a.s. that (i) $X_n\ge 0$ and $X_n\le X_{n+1}$ for all $n$ and (ii) $\lim\limits_{n\rightarrow\infty}X_n=X$, then
\[
\lim\limits_{n\rightarrow\infty}\condexpv{X_n}{\mathcal{G}}=\condexpv{X}{\mathcal{G}}\mbox{ a.s.}
\]
\item[(E12)] if (i) $|X_n|\le Y$ for all $n$ and (ii) $\lim\limits_{n\rightarrow\infty} X_n=X$, then
\[
\lim\limits_{n\rightarrow\infty}\condexpv{X_n}{\mathcal{G}}=\condexpv{X}{\mathcal{G}}\mbox{ a.s.}
\]
\item[(E13)] if $g:\Rset\rightarrow\Rset$ is a convex function and $\expv(|g(X)|)<\infty$, then $g(\condexpv{X}{\mathcal{G}})\le \condexpv{g(X)}{\mathcal{G}}$ a.s.
\end{compactitem}
We refer to~\cite[Chapter~9]{probabilitycambridge} for more details.

\section{Proofs for Martingale Results}\label{app:martingale}

The following version of Optional Stopping Theorem
is an extension of the one from \cite[Chapter 10]{probabilitycambridge}.
In the proof of the following theorem, for a stopping time $R$ and a non-negative integer $n\in\Nset_0$, we denote by $R\wedge n$ the random variable $\min\{R,n\}$.

\smallskip
\begin{theorem}[Optional Stopping Theorem
\protect\footnote{cf. \url{https://en.wikipedia.org/wiki/Optional_stopping_theorem}}
]~\label{thm:optstopping}
Consider any stopping time $R$ w.r.t a filtration $\{\mathcal{F}_n\}_{n\in\Nset_0}$ and any martingale (resp. supermartingale) $\{X_n\}_{n\in\Nset_0}$ adapted to $\{\mathcal{F}_n\}_{n\in\Nset_0}$.
Then $\expv\left(|Y|\right)<\infty$ and $\expv\left(Y\right)=\expv(X_0)$ (resp. $\expv\left(Y\right)\le\expv(X_0)$) if one of the following conditions hold:
\begin{compactenum}
\item there exists an $M\in (0,\infty)$ such that $\left|X_{R\wedge n}\right|<M$ a.s. for all $n\in\Nset_0$, and $Y=\lim\limits_{n\rightarrow\infty}X_{R\wedge n}$ a.s., where the existence of $Y$ follows from Doob's Convergence Theorem
    ;
\item $\expv(R)<\infty$, $Y=X_R$ and there exists a $c\in (0,\infty)$ such that for all $n\in\Nset_0$, $\condexpv{\left|X_{n+1}-X_n\right|}{\mathcal{F}_n}\le c$ a.s.
\end{compactenum}
Moreover,
\begin{compactitem}
\item[3)] if $\probm(R<\infty)=1$ and $X_n(\omega)\ge 0$ for all $n,\omega$, then  $\expv\left(X_R\right)\le\expv(X_0)$.
\end{compactitem}
\end{theorem}
\begin{proof}
The first item of the theorem follows directly from Dominated Convergence Theorem~\cite[Chapter 6.2]{probabilitycambridge} and properties for stopped processes (cf.~\cite[Chapter 10.9]{probabilitycambridge}).
The third item is from \cite[Chapter 10.10(d)]{probabilitycambridge}.
Below we prove the second item.
We have for every $n\in\Nset_0$,
\begin{eqnarray*}
\left|X_{R\wedge n}\right|&=& \left|X_0+\sum_{k=0}^{R\wedge n-1} \left(X_{k+1}-X_k\right)\right| \\
&=& \left|X_0+\sum_{k=0}^\infty \left(X_{k+1}-X_k\right)\cdot \mathbf{1}_{R>k\wedge n>k}\right|\\
&\le& \left|X_0\right|+\sum_{k=0}^\infty \left|\left(X_{k+1}-X_k\right)\cdot \mathbf{1}_{R>k\wedge n>k}\right| \\
&\le& \left|X_0\right|+\sum_{k=0}^\infty \left|\left(X_{k+1}-X_k\right)\cdot \mathbf{1}_{R>k}\right|\enskip. \\
\end{eqnarray*}
Note that
\begin{eqnarray*}
& &   \expv\left(\left|X_0\right|+\sum_{k=0}^\infty \left|\left(X_{k+1}-X_k\right)\cdot \mathbf{1}_{R>k}\right|\right) \\
&=& \mbox{(By Monotone Convergence Theorem~\cite[Chap. 6]{probabilitycambridge})} \\
& & \expv\left(\left|X_0\right|\right)+\sum_{k=0}^\infty \expv\left(\left|\left(X_{k+1}-X_k\right)\cdot \mathbf{1}_{R>k}\right|\right) \\
&=& \expv\left(\left|X_0\right|\right)+\sum_{k=0}^\infty \expv\left(\left|X_{k+1}-X_k\right|\cdot \mathbf{1}_{R>k}\right) \\
&=& \mbox{(By (E4))}\\
& & \expv\left(\left|X_0\right|\right)+\sum_{k=0}^\infty \expv\left(\condexpv{\left|X_{k+1}-X_k\right|\cdot \mathbf{1}_{R>k}}{\mathcal{F}_k}\right) \\
&=& \mbox{(by (E8))} \\
& &
\expv\left(\left|X_0\right|\right)+\sum_{k=0}^\infty \expv\left(\condexpv{\left|X_{k+1}-X_k\right|}{\mathcal{F}_k}\cdot \mathbf{1}_{R>k}\right) \\
&\le & \expv\left(\left|X_0\right|\right)+\sum_{k=0}^\infty \expv\left(c\cdot \mathbf{1}_{R>k}\right) \\
&=& \expv\left(\left|X_0\right|\right)+\sum_{k=0}^\infty c\cdot \probm\left(R>k\right) \\
&=& \expv\left(\left|X_0\right|\right)+\sum_{k=0}^\infty c\cdot \probm\left(k<R<\infty\right) \\
&=& \expv\left(\left|X_0\right|\right)+ c\cdot \expv(R) \\
&<& \infty\enskip.
\end{eqnarray*}
Thus, by Dominated Convergence Theorem~\cite[Chapter 6.2]{probabilitycambridge} and the fact that $X_R=\lim\limits_{n\rightarrow\infty} X_{R\wedge n}$ a.s.,
\[
\expv\left(X_R\right)=\expv\left(\lim\limits_{n\rightarrow\infty} X_{R\wedge n}\right)=\lim\limits_{n\rightarrow\infty}\expv\left(X_{R\wedge n}\right)\enskip.
\]
Then the result follows from properties for the stopped process $\{X_{R\wedge n}\}_{n\in\Nset_0}$ (cf.~\cite[Chapter 10.9]{probabilitycambridge}).
\end{proof}

In this section, for a random variable $R$ and a real number $M$, we denote by $R\wedge M$ the random variable $\min\{R,M\}$.

\noindent{\textbf{Proposition~\ref{thm:supm}} (Difference-bounded Supermartingales)\bf.}
Consider any difference-bounded supermartingale $\Gamma=\{X_n\}_{n\in\Nset_0}$ adapted to a filtration $\{\mathcal{F}_n\}_{n\in\Nset_0}$ satisfying the following conditions:
\begin{compactenum}
\item $X_0$ is a constant random variable;
\item for all $n\in\Nset_0$, it holds for all $\omega$ that (i) $X_n(\omega)\ge 0$ and (ii) $X_n(\omega)=0$ implies  $X_{n+1}(\omega)=0$;
\item there exists a $\delta\in(0,\infty)$ such that for all $n\in\Nset_0$, it holds a.s. that $X_n>0$ implies $\condexpv{|X_{n+1}-X_n|}{\mathcal{F}_n}\ge \delta$.
\end{compactenum}
Then $\probm(\stopping{\Gamma}<\infty)=1$ and the function $k\mapsto\probm\left(\stopping{\Gamma}\ge k\right)\in \mathcal{O}\left(\frac{1}{\sqrt{k}}\right)$.
\begin{proof}
The proof uses ideas from both~\cite[Chapter~10.12]{probabilitycambridge} and \cite[Theorem 4.1]{DBLP:journals/jcss/BrazdilKKV15}.
Let $c\in (0,\infty)$ be such that for every $n\in\Nset_0$, it holds a.s. that $|X_{n+1}-X_n|\le c$.
Let $\delta$ be given as in the statement of the theorem.
W.l.o.g, we assume that $X_0>0$.
Note that from (E13), it holds a.s. that $X_n>0$ implies
\begin{equation}\label{eq:proof:convexity}
\condexpv{{(X_{n+1}-X_n)}^2}{\mathcal{F}_n}\ge {(\condexpv{|X_{n+1}-X_n|}{\mathcal{F}_n})}^2\ge \delta^2\enskip.
\end{equation}
Fix any sufficiently small real number $t\in (0,\infty)$ such that
\[
e^{c\cdot t}-(1+c\cdot t +\frac{1}{2}\cdot c^2\cdot t^2)\left(=\sum_{j=3}^\infty\frac{(c\cdot t)^j}{{j}{!}} \right)\le \frac{\delta^2}{4}\cdot t^2\enskip.
\]
Define the discrete-time stochastic process $\{Y_n\}_{n\in\Nset_0}$ by
\begin{equation}\label{eq:proof4:definition}
Y_n:=\frac{e^{-t\cdot X_n}}{\prod_{j=0}^{n-1} \condexpv{e^{-t\cdot \left(X_{j+1}-X_{j}\right)}}{\mathcal{F}_j}}\enskip.
\end{equation}
Note that from difference-boundedness, $0<Y_n\le e^{n\cdot c\cdot t}$ a.s. for all $n\in\Nset_0$.
Then the followings hold a.s.:
\begin{eqnarray}\label{eq:proof4:martingale}
& & \condexpv{Y_{n+1}}{\mathcal{F}_n} \nonumber\\
&=& \condexpv{\frac{e^{-t\cdot X_{n+1}}}{\prod_{j=0}^{n} \condexpv{e^{-t\cdot \left(X_{j+1}-X_{j}\right)}}{\mathcal{F}_j}}}{\mathcal{F}_n} \nonumber\\
&=& \condexpv{\frac{e^{-t\cdot X_n}\cdot e^{-t\cdot \left(X_{n+1}-X_{n}\right)}}{\prod_{j=0}^{n} \condexpv{e^{-t\cdot \left(X_{j+1}-X_{j}\right)}}{\mathcal{F}_j}}}{\mathcal{F}_n} \nonumber\\
&=& \mbox{\S~By (E8), (E1) \S} \nonumber\\
& &\frac{e^{-t\cdot X_n}\cdot \condexpv{ e^{-t\cdot \left(X_{n+1}-X_{n}\right)}}{\mathcal{F}_n}}{\prod_{j=0}^{n} \condexpv{e^{-t\cdot \left(X_{j+1}-X_{j}\right)}}{\mathcal{F}_j}} \nonumber\\
&=& \frac{e^{-t\cdot X_n}}{\prod_{j=0}^{n-1} \condexpv{e^{-t\cdot \left(X_{j+1}-X_{j}\right)}}{\mathcal{F}_j}} \nonumber\\
&=& Y_n\enskip.
\end{eqnarray}
Hence, $\{Y_n\}_{n\in\Nset_0}$ is a martingale.
For every $n\in\Nset_0$, it holds a.s. that $X_n>0$ implies
\begin{eqnarray}\label{eq:proof4:exponentiation}
& &\expv\left(e^{-t\cdot \left(X_{n+1}-X_{n}\right)}\mid\mathcal{F}_n\right) \nonumber\\
&=&\expv\left(\sum_{j=0}^{\infty} \frac{(-1)^j\cdot t^j\cdot {\left(X_{n+1}-X_{n}\right)}^j}{{j}{!}}\mid\mathcal{F}_n\right) \nonumber\\
&=& \mbox{\S~By (E12) \S}\nonumber\\
& &
\sum_{j=0}^{\infty}\expv\left(\frac{(-1)^j\cdot t^j\cdot {\left(X_{n+1}-X_{n}\right)}^j}{{j}{!}}\mid\mathcal{F}_n\right) \nonumber\\
&=& \mbox{\S~By (E6) \S} \nonumber\\
& &
1-t\cdot \condexpv{X_{n+1}-X_{n}}{\mathcal{F}_n}+\frac{t^2}{2}\cdot\condexpv{(X_{n+1}-X_{n})^2}{\mathcal{F}_n}+\sum_{j=3}^{\infty}\condexpv{\frac{(-1)^j\cdot t^j\cdot {\left(X_{n+1}-X_{n}\right)}^j}{{j}{!}}}{\mathcal{F}_n}\nonumber\\
&\ge & 1+\frac{t^2}{2}\cdot\condexpv{(X_{n+1}-X_{n})^2}{\mathcal{F}_n}- \sum_{j=3}^{\infty}\frac{{(c\cdot t)}^j}{{j}{!}}~~\\
&\ge & 1+\frac{\delta^2}{4}\cdot t^2\enskip.\nonumber
\end{eqnarray}

Thus,
\begin{compactitem}
\item $|Y_{\stopping{\Gamma}\wedge n}|\le 1$ a.s. for all $n\in\Nset_0$, and
\item it holds a.s. that
\begin{equation}\label{eq:thmmartingale}
\left(\lim\limits_{n\rightarrow\infty} Y_{n\wedge \stopping{\Gamma}}\right)(\omega)=
\begin{cases}
0 & \mbox{if }\stopping{\Gamma}(\omega)=\infty\\
Y_{\stopping{\Gamma}(\omega)}(\omega) & \mbox{if }\stopping{\Gamma}(\omega)<\infty
\end{cases}\enskip.
\end{equation}
\end{compactitem}
Then from Optional Stopping Theorem (Item 1 of Theorem~\ref{thm:optstopping}), by letting $Y_\infty:=\lim\limits_{n\rightarrow\infty} Y_{n\wedge\stopping{\Gamma}}$
one has that
\[
\expv\left(Y_\infty\right)=\expv\left(Y_0\right)=e^{-t\cdot \expv(X_0)}\enskip.
\]
Moreover, from~(\ref{eq:thmmartingale}), one can obtain that
\begin{eqnarray}\label{eq:proof4:derivation}
& &\expv\left(Y_\infty\right)\nonumber\\
&=& \mbox{\S~By Definition \S} \nonumber\\
& & \int Y_\infty\,\mathrm{d}\probm \nonumber\\
&=& \mbox{\S~By Linear Property of Lebesgue Integral \S} \nonumber\\
& &\int Y_\infty\cdot\mathbf{1}_{\stopping{\Gamma}=\infty}\,\mathrm{d}\probm+\int Y_\infty\cdot\mathbf{1}_{\stopping{\Gamma}<\infty}\,\mathrm{d}\probm\nonumber\\
&=& \mbox{\S~By Monotone Convergence Theorem~\cite[Chap. 6]{probabilitycambridge} \S} \nonumber\\
& & 0\cdot\probm\left(\stopping{\Gamma}=\infty\right)+\sum_{n=0}^{\infty}\int Y_\infty\cdot\mathbf{1}_{\stopping{\Gamma}=n}\,\mathrm{d}\probm\nonumber\\
&=& \sum_{n=0}^{\infty}\int Y_n\cdot\mathbf{1}_{\stopping{\Gamma}=n}\,\mathrm{d}\probm\nonumber\\
&\le& \mbox{\S~By (\ref{eq:proof4:exponentiation}) and $X_{n}\ge 0$ \S} \nonumber\\
& & \sum_{n=0}^{\infty}\int {\left(1+\frac{\delta^2}{4}\cdot t^2\right)}^{-n} \cdot\mathbf{1}_{\stopping{\Gamma}=n}\,\mathrm{d}\probm\nonumber\\
&=& \sum_{n=0}^{\infty} {\left(1+\frac{\delta^2}{4}\cdot t^2\right)}^{-n} \cdot\probm\left(\stopping{\Gamma}=n\right) \nonumber\\
&= & \sum_{n=0}^{k-1} {\left(1+\frac{\delta^2}{4}\cdot t^2\right)}^{-n} \cdot\probm\left(\stopping{\Gamma}=n\right)+\sum_{n=k}^{\infty} {\left(1+\frac{\delta^2}{4}\cdot t^2\right)}^{-n} \cdot\probm\left(\stopping{\Gamma}=n\right)\nonumber\\
&\le & \left(1-\probm\left(\stopping{\Gamma}\ge k\right)\right)+{\left(1+\frac{\delta^2}{4}\cdot t^2\right)}^{-k}\cdot \probm\left(\stopping{\Gamma}\ge k\right)~~
\end{eqnarray}
for any $k\in\Nset$.
It follows that for all $k\in\Nset$,
\[
e^{-t\cdot\expv(X_0)}\le 1-\left(1-\left(1+\frac{\delta^2}{4}\cdot t^2\right)^{-k}\right)\cdot \probm\left( \stopping{\Gamma}\ge k\right)\enskip.
\]
Hence, for any $k\in\Nset$ and sufficiently small $t\in (0,\infty)$,
\[
\probm\left(\stopping{\Gamma}\ge k\right)\le \frac{1-e^{-t\cdot\expv(X_0)}}{1-\left(1+\frac{\delta^2}{4}\cdot t^2\right)^{-k}}\enskip.
\]
Then for sufficiently large $k\in\Nset$ with $t:=\frac{1}{\sqrt{k}}$,
\[
\probm\left(\stopping{\Gamma}\ge k\right)\le \frac{1-e^{-\frac{\expv(X_0)}{\sqrt{k}}}}{1-\left(1+\frac{\delta^2}{4}\cdot \frac{1}{k}\right)^{-k}}\enskip.
\]
Using the facts that $\lim\limits_{k\rightarrow\infty}(1+\frac{\delta^2}{4}\cdot\frac{1}{k})^{k}= e^{\frac{\delta^2}{4}}$ and $\lim\limits_{z\rightarrow 0^+}\frac{1-e^{-z}}{z}=1$, we have that the function
\[
k\mapsto\probm\left(\stopping{\Gamma}\ge k\right)\in \mathcal{O}\left(\frac{1}{\sqrt{k}}\right)\enskip.
\]
Since $\probm(\stopping{\Gamma}=\infty)=\lim\limits_{k\rightarrow\infty}\probm\left(\stopping{\Gamma}\ge k\right)$, one obtains immediately that $\probm(\stopping{\Gamma}=\infty)=0$ and $\probm(\stopping{\Gamma}<\infty)=1$.
\end{proof}

\noindent\textbf{Example~\ref{ex:special:nonnegativity}.}
In Proposition~\ref{thm:supm}, the Non-negativity condition is necessary; in other words,  it is necessary having  $X_{\stopping{\Gamma}}=0$ rather than $X_{\stopping{\Gamma}}\le 0$ when $\stopping{\Gamma}<\infty$.
This can be observed as follows.
Consider the discrete-time stochastic processes $\{X_n\}_{n\in\Nset_0}$ and $\Gamma=\{Y_n\}_{n\in\Nset_0}$ given as follows:
\begin{compactitem}
\item the random variables $X_0,\dots,X_n,\dots$ are independent, $X_0$ is the random variable with constant value $\frac{1}{2}$ and each $X_n$ ($n\ge 1$) satisfies that $\probm\left(X_n=1\right)=e^{-\frac{1}{n^2}}$ and $\probm\left(X_n=-4\cdot n^2\right)=1-e^{-\frac{1}{n^2}}$;
\item $Y_n:=\sum_{j=0}^{n}X_j$ for $n\ge 0$.
\end{compactitem}
Let the filtration $\{\mathcal{F}_n\}_{n\in\Nset_0}$ be given such that each $\mathcal{F}_n$ is the $\sigma$-algebra generated by $X_0,\dots,X_n$ (i.e., the smallest $\sigma$-algebra that makes  $X_0,\dots,X_n$ measurable).
It is straightforward to see that every $Y_n$ is integrable and $\mathcal{F}_n$-measurable, and every $X_{n+1}$ is independent of $\mathcal{F}_n$.
Thus, for $n\ge 0$
we have that (cf. properties for conditional expectation in Appendix~\ref{app:condexpv})
\begin{eqnarray*}
\condexpv{Y_{n+1}}{\mathcal{F}_n}&=& \condexpv{Y_{n}+X_{n+1}}{\mathcal{F}_n}\\
\mbox{(by (E6), (E5))} &=& Y_{n}+\condexpv{X_{n+1}}{\mathcal{F}_n}\\
\mbox{(by (E9))} &=& Y_{n}+\expv\left(X_{n+1}\right)\\
&=& Y_{n}+\left(e^{-\frac{1}{(n+1)^2}} - 4\cdot \frac{1-e^{-\frac{1}{(n+1)^2}}}{\frac{1}{(n+1)^2}}\right)\\
&\le& Y_{n}+1-4\cdot \left(1-e^{-1}\right)\\
&\le& Y_n-1.52~~~,
\end{eqnarray*}
where the first inequality is obtain by the fact that the function $x\mapsto \frac{1-e^{-x}}{x}$ is decreasing over $(0,\infty)$. Hence, $\{Y_n\}_{n\in\Nset_0}$ satisfies even the ranking condition for ranking supermartingales, thus satisfying the LBCAD condition.
However, since $Y_n<0$ once $X_n=-4\cdot n^2$, 
$\probm\left(\stopping{\Gamma}> n\right)=\prod_{j=1}^{n}e^{-\frac{1}{j^2}}$.
It follows directly that
$\probm\left(\stopping{\Gamma}=\infty\right)=\lim\limits_{n\rightarrow\infty}\probm\left(\stopping{\Gamma}> n\right)=e^{-\frac{\pi^2}{6}}>0$\enskip.

\noindent\textbf{Theorem~\ref{thm:soundness}} (Soundness)\textbf{.}
If there exists a difference-bounded supermartingale map $h$ for $P$, then for any initial valuation $\pv_0$
we have $\probm(T< \infty)=1$ and $k\mapsto\probm(T\geq k)\in \mathcal{O}\left(\frac{1}{\sqrt{k}}\right)$.

\begin{proof}
  Let $h$ be any  difference-bounded supermartingale map $h$ for the single while loop program $P$, $\pv$ be any initial valuation and $\delta, \zeta$ be the parameters in Definition~\ref{def:supmmap}.

  If $\pv\not\models\phi$, then the program $P$ terminates and we are done. Now we suppose $\pv\models\phi$.
  Let $\mathcal{M}=(S,\textbf{P})$ be the Markov chain defining the semantics of $P$, and let $\{Y_n=(\ell_n,\pv_n)\}\}_{n\in\Nset_0}$ be the stochastic process of it.
  Now we define the stochastic process $\Gamma=\{X_n\}_{n\in\Nset_0}$ adapted to $\{\mathcal{F}_n\}_{n\in\Nset_0}$  by
  $X_n=h(\ell_n,\pv_n)$.

  By Definition~\ref{def:supmmap}, we have $X_0\geq\delta>0$ by (D2.1). Moreover, suppose $\ell_{n-1}=\verb"in"$ and $\pv_{n-1}\models\phi$,  by applying (D2.2) we have $X_n\geq\delta>0$. Now we check the conditions in Proposition~\ref{thm:supm}.
  \begin{compactenum}
    \item $X_0=h(\verb"in",\pv)$ is a constant random variable
    \item $\Gamma=\{X_n\}_{n\in\Nset_0}$ is a difference-bounded supermartingale by (D3.1) and (D4). We prove it as follows.\\
           We first prove that $\expv(|X_{n}|)<\infty$ for all $n$ by induction on $n$. $\expv(X_{0})=h(\verb"in",\pv)<\infty$. Suppose that $\expv(|X_{n}|)=\expv(X_{n})<\infty$ for every $n\leq k$, now we consider $k+1$.
            For any $Y_k=(\ell,\pv)$ such that $\ell=\verb"in"$ and $\pv\models\phi$, we have
           $\expv(X_{k+1}|Y_k=(\ell,\pv))= \Sigma_{\rv\in \Zset^{|R|}} \bar{\Upsilon}(\rv)\cdot h(\verb"in",F(\pv,\rv))\leq h (\ell,\pv)=X_k<\infty $; for any $Y_k=(\verb"out",\pv)$ we have $\expv(X_{k+1}|Y_k=(\ell,\pv))= 0<\delta\leq X_k$; for any $Y_k=(\ell,\pv)$ such that $\ell=\verb"in"$ and $\pv\not\models\phi$ we have $\expv(X_{k+1}|Y_k=(\ell,\pv))= 0<\delta\leq X_k$. It tails that $\expv(X_{k+1})=\expv(\expv(X_{k+1}|Y_k=(\ell,\pv)))\leq \expv(X_k)<\infty$. Now  we have $\expv(|X_{n}|)=\expv(X_{n})\leq \expv(X_{0})<\infty$ for all $n$.\\
           Moreover, since $\expv(X_{k+1}|Y_k=(\ell,\pv))\leq X_k$ for all $Y_k$, we have $\expv(X_{k+1}|\mathcal{F}_n)\leq X_k$ a.s..
            We have proved that $\Gamma$ is a supermartingale. The condition (D4) guarantees that $\Gamma$ is difference bounded.

    \item for all $n\in\Nset_0$, it holds for all $\omega$ that \\
    (i) $X_n(\omega)\ge 0$. \\
        (ii) If $X_n=0$, then $\ell_n=\verb"out"$   which implies $\ell_{n+1}=\verb"out"$ and $X_{n+1}=0$ (D1, D2).
    \item  Let $X_n>0$, we consider the cases below.
        If $\ell_n=\verb"in"$ and  $\pv_n\not\models\phi$, then $\condexpv{|X_{n+1}-X_{n}|}{Y_n}= X_n\geq\delta$  by (D2).\\
        If $\ell_n=\verb"in"$ and $\pv_n\models\phi$, then  $\condexpv{|X_{n+1}-X_n|}{Y_n}={\Sigma_{\rv\in \Zset^{|R|}}\bar{\Upsilon}(\rv)\cdot |g(\rv)|}\geq\delta $ by (D3.2).
        It follows that $\condexpv{|X_{n+1}-X_n|}{\mathcal{F}_n}\geq\delta$ a.s..
  \end{compactenum}
   It follows that $\Gamma=\{X_n\}_{n\in\Nset_0}$ satisfy the conditions of Proposition~\ref{thm:supm}, and we have $\probm(Z_\Gamma<\infty)=1$ and $k\mapsto\probm\left(\stopping{\Gamma}\ge k\right)\in \mathcal{O}\left(\frac{1}{\sqrt{k}}\right)$. $\Gamma$ stops implies the program $P$ terminates with the same tail bound, and we are done.
\end{proof}

\noindent\textbf{Proposition~\ref{thm:supmextended}}(General Supermartingales)\textbf{.}
Consider any supermartingale $\Gamma=\{X_n\}_{n\in\Nset_0}$ adapted to a filtration $\{\mathcal{F}_n\}_{n\in\Nset_0}$ satisfying the following conditions:
\begin{compactenum}
\item $X_0$ is a constant random variable;
\item for all $n\in\Nset_0$, it holds for all $\omega$ that (i) $X_n(\omega)\ge 0$ and (ii) $X_n(\omega)=0$ implies  $X_{n+1}(\omega)=0$;
\item there exists a $\delta\in(0,\infty)$ such that for all $n\in\Nset_0$, it holds a.s. that $\condexpv{|X_{n+1}-X_n|}{\mathcal{F}_n}\ge \delta\cdot\mathbf{1}_{X_n>0}$.
\end{compactenum}
Then $\probm(\stopping{\Gamma}<\infty)=1$ and the function $k\mapsto\probm\left(\stopping{\Gamma}\ge k\right)\in \mathcal{O}\left(k^{-\frac{1}{6}}\right)$.
\begin{proof}
W.l.o.g., we assume that $X_0>0$. Let $\delta$ be given as in the statement of the theorem.
From
\[
\lim\limits_{k\rightarrow\infty}\frac{1-e^{-\frac{\expv(X_0)}{\sqrt{k}}}}{\frac{\expv(X_0)}{\sqrt{k}}}=1\mbox{ and } \lim\limits_{k\rightarrow\infty} \left(1+\frac{\delta^2}{16}\cdot \frac{1}{k}\right)^{-k}= e^{-\frac{\delta^2}{16}}\enskip,
\]
one can fix a constant natural number $N$ such that for all $k\ge N$,
\[
\frac{1-e^{-\frac{\expv(X_0)}{\sqrt{k}}}}{\frac{\expv(X_0)}{\sqrt{k}}}\le \frac{3}{2}\mbox{ and } 1-\left(1+\frac{\delta^2}{16}\cdot \frac{1}{k}\right)^{-k}\ge \frac{1-e^{-\frac{\delta^2}{16}}}{2}\enskip.
\]
Let
\[
C:= \frac{3}{2}\cdot \expv(X_0)\cdot \frac{2}{1-e^{-\frac{\delta^2}{16}}}\enskip.
\]
Choose a constant $c\in (0,1)$ such that
\[
\sum_{j=3}^{\infty} \frac{c^{j-2}}{j!}\le \frac{\delta^2}{16}\enskip.
\]
Note that from (E6), it holds a.s. for all $n$ that
\begin{eqnarray*}
& &\condexpv{|X_{n+1}-X_n|}{\mathcal{F}_n}\\
&=&\condexpv{\mathbf{1}_{X_{n+1}< X_n}\cdot (X_{n}-X_{n+1})}{\mathcal{F}_n}\\
& &{}+\condexpv{\mathbf{1}_{X_{n+1}\ge X_n}\cdot (X_{n+1}-X_{n})}{\mathcal{F}_n}\\
&\ge& \mathbf{1}_{X_n>0}\cdot \delta\enskip.
\end{eqnarray*}
Moreover, from (E5), (E6) and definition of supermartingales, it holds a.s. that
\begin{eqnarray*}
& &\condexpv{X_{n+1}-X_n}{\mathcal{F}_n}\\
&=&\condexpv{\mathbf{1}_{X_{n+1}< X_n}\cdot (X_{n+1}-X_{n})}{\mathcal{F}_n}\\
& &{}+\condexpv{\mathbf{1}_{X_{n+1}\ge X_n}\cdot (X_{n+1}-X_{n})}{\mathcal{F}_n}\\
&=&-\condexpv{\mathbf{1}_{X_{n+1}< X_n}\cdot (X_{n}-X_{n+1})}{\mathcal{F}_n}\\
& &{}+\condexpv{\mathbf{1}_{X_{n+1}\ge X_n}\cdot (X_{n+1}-X_{n})}{\mathcal{F}_n}\\
&\le & 0\enskip.
\end{eqnarray*}
It follows that for all $n$, it holds a.s. that
\begin{equation}\label{eq:supmextended:delta}
\condexpv{\mathbf{1}_{X_{n+1}< X_n}\cdot (X_{n}-X_{n+1})}{\mathcal{F}_n}\ge \mathbf{1}_{X_n>0}\cdot \frac{\delta}{2}\enskip.
\end{equation}

Let $M$ be any real number satisfying $M> \max\{\expv(X_0), \sqrt[6]{N}\}$ and define the stopping time $R_M$ w.r.t $\{\mathcal{F}_n\}_{n\in\Nset_0}$ by
\[
R_M(\omega):=  \min\{n\mid X_{n}(\omega)\le 0\mbox{ or } X_{n}(\omega)\ge M\}
\]
where $\min\emptyset:=\infty$. Define the stochastic process $\Gamma'=\{X'_n\}_{n\in\Nset_0}$ adapted to $\{\mathcal{F}_n\}_{n\in\Nset_0}$ by:
\begin{equation}\label{eq:proof:primesupermartingale}
X'_n=X_n\wedge M~\mbox{ for all } n\in\Nset_0\enskip.
\end{equation}
It is clear that $\Gamma'$ is difference-bounded.
Below we prove that $\Gamma'$ is a supermartingale.
This can be observed from the following:
\begin{eqnarray*}
& &\condexpv{X'_{n+1}}{\mathcal{F}_n}-X'_n \\
&=& \mbox{\S~By~(E5), (E6) \S} \\
& &\condexpv{X'_{n+1}-X'_n}{\mathcal{F}_n} \\
&=& \mbox{\S~By~(E6) \S}\\
& & \condexpv{\mathbf{1}_{X_n>M}\cdot \left((X_{n+1}\wedge M) - M\right)}{\mathcal{F}_n} \\
& & \quad{}+\condexpv{\mathbf{1}_{X_n\le M}\cdot \left((X_{n+1}\wedge M)-X_n\right)}{\mathcal{F}_n}\\
&\le & \mbox{\S~By~(E8), (E10) \S}\\
& & \mathbf{1}_{X_n\le M}\cdot\condexpv{X_{n+1}-X_n}{\mathcal{F}_n}\\
&\le& 0\enskip.
\end{eqnarray*}
Hence $\Gamma'$ is a difference-bounded supermartingale.
Moreover, we have that the followings hold a.s. for all $n$:
\begin{eqnarray*}
& &\mathbf{1}_{0<X'_n< M}\cdot \condexpv{|X'_{n+1}-X'_n|}{\mathcal{F}_n} \\
&=& \mbox{\S~By~(E8) \S} \\
& & \condexpv{\mathbf{1}_{0<X'_n< M}\cdot |X'_{n+1}-X'_n|}{\mathcal{F}_n}\\
&\ge & \mbox{\S~By~(E10), (E6) \S} \\
& & \condexpv{\mathbf{1}_{0<X'_n< M}\cdot \mathbf{1}_{X'_{n+1}<X'_n}\cdot \left(X'_{n}-X'_{n+1}\right)}{\mathcal{F}_n} \\
&=& \condexpv{\mathbf{1}_{0<X'_n< M}\cdot \mathbf{1}_{X_{n+1}<X_n}\cdot \left(X_{n}-X_{n+1}\right)}{\mathcal{F}_n}\\
&=& \mbox{\S~By~(E8) \S} \\
& & \mathbf{1}_{0<X'_n< M}\cdot\condexpv{ \mathbf{1}_{X_{n+1}<X_n}\cdot \left(X_{n}-X_{n+1}\right)}{\mathcal{F}_n}\\
&\ge & \mbox{\S~By~(\ref{eq:supmextended:delta}) \S} \\
& & \mathbf{1}_{0<X'_n< M}\cdot \mathbf{1}_{X_n>0}\cdot \frac{\delta}{2}\\
&=& \mathbf{1}_{0<X'_n< M}\cdot \frac{\delta}{2}\enskip.
\end{eqnarray*}
Since $\condexpv{|X'_{n+1}-X'_n|}{\mathcal{F}_n}\ge 0$ a.s. (from (E10)), we obtain that a.s.
\[
\condexpv{|X'_{n+1}-X'_n|}{\mathcal{F}_n}\ge \mathbf{1}_{0<X'_n< M}\cdot \frac{\delta}{2}\enskip.
\]
Hence, from (E13), it holds a.s. for all $n$ that
\begin{eqnarray*}
\condexpv{(X'_{n+1}-X'_n)^2}{\mathcal{F}_n}&\ge &{\left(\condexpv{|X'_{n+1}-X'_n|}{\mathcal{F}_n}\right)}^2\\
&\ge& \mathbf{1}_{0<X'_n< M}\cdot \frac{\delta^2}{4}\enskip.
\end{eqnarray*}

Now define the discrete-time stochastic process $\{Y_n\}_{n\in\Nset_0}$ by
\[
Y_n:=\frac{e^{-t\cdot X'_n}}{\prod_{j=0}^{n-1} \condexpv{e^{-t\cdot \left(X'_{j+1}-X'_{j}\right)}}{\mathcal{F}_j}}
\]
where $t$ is an arbitrary real number in $(0,\frac{c}{M^3}]$.
Note that from difference-boundedness and (E10), $0<Y_n\le e^{n\cdot M\cdot t}$ a.s. for all $n\in\Nset_0$.
Then by the same analysis in (\ref{eq:proof4:martingale}), $\{Y_n\}_{n\in\Nset_0}$ is a martingale.
Furthermore, by similar analysis in (\ref{eq:proof4:exponentiation}), one can obtain that for every $n$, it holds a.s. that $0<X'_n< M$ implies
\begin{eqnarray*}
& &\expv\left(e^{-t\cdot \left(X'_{n+1}-X'_{n}\right)}\mid\mathcal{F}_n\right) \nonumber\\
&\ge & 1+\frac{t^2}{2}\cdot\condexpv{(X'_{n+1}-X'_{n})^2}{\mathcal{F}_n}- \sum_{j=3}^{\infty}\frac{{(M\cdot t)}^j}{{j}{!}}~~\\
&\ge & 1+\frac{t^2}{2}\cdot\frac{\delta^2}{4}- t^2\cdot \sum_{j=3}^{\infty}\frac{M^j\cdot t^{j-2}}{{j}{!}}~~\\
&\ge & 1+\frac{\delta^2}{8}\cdot t^2- t^2\cdot \sum_{j=3}^{\infty}\frac{M^{-2\cdot j+6}\cdot c^{j-2}}{{j}{!}}~~\\
&\ge & 1+\frac{\delta^2}{8}\cdot t^2- t^2\cdot \sum_{j=3}^{\infty}\frac{c^{j-2}}{{j}{!}}~~\\
&\ge & 1+\frac{\delta^2}{8}\cdot t^2- t^2\cdot \frac{\delta^2}{16}~~\\
&\ge & 1+\frac{\delta^2}{16}\cdot t^2\enskip.\nonumber
\end{eqnarray*}

Thus,
\begin{compactitem}
\item $|Y_{R_M\wedge n}|\le 1$ a.s. for all $n\in\Nset_0$, and
\item it holds a.s. that
\[
\left(\lim\limits_{n\rightarrow\infty} Y_{n\wedge R_M}\right)(\omega)=
\begin{cases}
0 & \mbox{if }R_M(\omega)=\infty\\
Y_{R_M(\omega)}(\omega) & \mbox{if }R_M(\omega)<\infty
\end{cases}\enskip.
\]
\end{compactitem}
Then from Optional Stopping Theorem (Item 1 of Theorem~\ref{thm:optstopping}), by letting $Y_\infty:=\lim\limits_{n\rightarrow\infty} Y_{n\wedge R_M}$
one has that
\[
\expv\left(Y_\infty\right)=\expv\left(Y_0\right)=e^{-t\cdot \expv(X_0)}\enskip.
\]
Moreover, one can obtain that
\begin{eqnarray}\label{eq:proof4:derivation2}
& &\expv\left(Y_\infty\right)\nonumber\\
&=& \mbox{\S~By Definition \S} \nonumber\\
& & \int Y_\infty\,\mathrm{d}\probm \nonumber\\
&=& \mbox{\S~By Linear Property of Lebesgue Integral \S} \nonumber\\
& &\int Y_\infty\cdot\mathbf{1}_{R_M=\infty}\,\mathrm{d}\probm+\int Y_\infty\cdot\mathbf{1}_{R_M<\infty}\,\mathrm{d}\probm\nonumber\\
&=& \mbox{\S~By Monotone Convergence Theorem~\cite[Chap. 6]{probabilitycambridge} \S} \nonumber\\
& & 0\cdot\probm\left(R_M=\infty\right)+\sum_{n=0}^{\infty}\int Y_\infty\cdot\mathbf{1}_{R_M=n}\,\mathrm{d}\probm\nonumber\\
&=& \sum_{n=0}^{\infty}\int Y_n\cdot\mathbf{1}_{R_M=n}\,\mathrm{d}\probm\nonumber\\
&\le& \mbox{\S~By $X'_{n}\ge 0$ \S} \nonumber\\
& & \sum_{n=0}^{\infty}\int {\left(1+\frac{\delta^2}{16}\cdot t^2\right)}^{-n} \cdot\mathbf{1}_{R_M=n}\,\mathrm{d}\probm\nonumber\\
&=& \sum_{n=0}^{\infty} {\left(1+\frac{\delta^2}{16}\cdot t^2\right)}^{-n} \cdot\probm\left(R_M=n\right) \nonumber\\
&= & \sum_{n=0}^{k-1} {\left(1+\frac{\delta^2}{16}\cdot t^2\right)}^{-n} \cdot\probm\left(R_M=n\right)+\sum_{n=k}^{\infty} {\left(1+\frac{\delta^2}{16}\cdot t^2\right)}^{-n} \cdot\probm\left(R_M=n\right)\nonumber\\
&\le & \left(1-\probm\left(R_M\ge k\right)\right)+{\left(1+\frac{\delta^2}{16}\cdot t^2\right)}^{-k}\cdot \probm\left(R_M\ge k\right)~~
\end{eqnarray}
for any $k\in\Nset$.
It follows that for all $k\in\Nset$,
\[
e^{-t\cdot\expv(X_0)}\le 1-\left(1-\left(1+\frac{\delta^2}{16}\cdot t^2\right)^{-k}\right)\cdot \probm\left( R_M\ge k\right)\enskip.
\]
Hence, for any $k\in\Nset$ and $t\in \left(0,\frac{c}{M^3}\right]$,
\[
\probm\left(R_M\ge k\right)\le \frac{1-e^{-t\cdot\expv(X_0)}}{1-\left(1+\frac{\delta^2}{16}\cdot t^2\right)^{-k}}\enskip.
\]
Then for any natural number $k\ge\frac{M^6}{c^2}$ with $t:=\frac{1}{\sqrt{k}}$,
\[
\probm\left(R_M\ge k\right)\le \frac{1-e^{-\frac{\expv(X_0)}{\sqrt{k}}}}{1-\left(1+\frac{\delta^2}{16}\cdot \frac{1}{k}\right)^{-k}}\enskip.
\]
In particular, we have that for all natural numbers $k\ge\frac{M^6}{c^2}$,
\[
\probm\left(R_M\ge k\right)\le C\cdot\frac{1}{\sqrt{k}}\enskip.
\]
Since $\probm(R_M=\infty)=\lim\limits_{k\rightarrow\infty}\probm\left(R_M\ge k\right)$, one obtains that $\probm(R_M=\infty)=0$ and $\probm(R_M<\infty)=1$.
By applying the third item of Optional Stopping Theorem (cf. Theorem~\ref{thm:optstopping}), one has that
$\expv(X_{R_M})\le \expv(X_0)$.
Thus, by Markov's Inequality,
\[
\probm(X_{R_M}\ge M)\le \frac{\expv(X_{R_M})}{M}\le \frac{\expv(X_0)}{M}\enskip.
\]

Now for any natural number $k$ such that $M:=\sqrt[6]{c^2\cdot k}>\max\{\expv(X_0),\sqrt[6]{N}\}$, we have
\begin{eqnarray*}
& &\probm(\stopping{\Gamma}\ge k) \\
&=&\probm(\stopping{\Gamma}\ge k\wedge X_{R_M}=0) + \probm(\stopping{\Gamma}\ge k\wedge X_{R_M}\ge M)\\
&=&\probm(R_M\ge k\wedge X_{R_M}=0) + \probm(\stopping{\Gamma}\ge k\wedge X_{R_M}\ge M)\\
&\le& \probm(R_M\ge k) + \probm(X_{R_M}\ge M)\\
&\le& \frac{C}{\sqrt{k}}+\frac{\expv(X_0)}{M} \\
&=& \frac{C}{\sqrt{k}}+\frac{\expv(X_0)}{\sqrt[6]{c^2\cdot k}}\enskip.
\end{eqnarray*}
It follows that $\probm(\stopping{\Gamma}<\infty)=1$ and $k\mapsto\probm\left(\stopping{\Gamma}\ge k\right)\in \mathcal{O}\left(k^{-\frac{1}{6}}\right)$.
\end{proof}

\end{document}